\documentclass[ijds]{informs4} % ,blindrev

\OneAndAHalfSpacedXI

\usepackage{amsmath,amssymb,amsfonts}
\usepackage{tikz-cd}
\usepackage{comment}

\usepackage{natbib}
 \bibpunct[, ]{(}{)}{,}{a}{}{,}%
 \def\bibfont{\small}%

\usepackage{rotating}
\usepackage{fancyvrb}
\usepackage{subcaption} % Not in original template

\newcommand{\e}{{\bf 1}}
\newcommand{\vecn}{{\bf 0}}
\newcommand{\tS}[1]{\tilde S_{#1}(-\theta_Z)}

\newtheorem{cor}{Corollary}
\newtheorem{prop}{Proposition}

\TheoremsNumberedThrough     % Preferred (Theorem 1, Lemma 1, Theorem 2)
\ECRepeatTheorems
\JOURNAL{Operations Research}

\EquationsNumberedThrough    % Default: (1), (2), ...
\MANUSCRIPTNO{xxx}

\begin{document}
%\setcounter{page}{1}
%\DOI{10.1287/xxxx.0000.0000}%


\RUNTITLE{Nudge*(M) Scheduling: tails and performance}

\TITLE{Tail Optimality and Performance Analysis of the Nudge*(M) Scheduling Algorithm }

\ARTICLEAUTHORS{%
\AUTHOR{Nils Charlet}
%,\textsuperscript{a} Second Author,\textsuperscript{b} Third Author,\textsuperscript{c} Fourth Author,\textsuperscript{c}

\AFF{Dept. Computer Science, University of Antwerp, \EMAIL{nils.charlet@uantwerpen.be}}
%\textsuperscript{b}School of Industrial Engineering, Good College, Collegeville, Maine 01234 \EMAIL{secauth@goodcoll.edu}; 
%\textsuperscript{c}Their Common Affiliation \EMAIL{thauth@anywhere.edu, fourauth@anywhere.edu}

%mirko.janc@informs.org
\AUTHOR{Benny Van Houdt}

\AFF{Dept. Computer Science, University of Antwerp, \EMAIL{benny.vanhoudt@uantwerpen.be}}
}

\ABSTRACT{
Recently it was shown that the response time of First-Come-First-Served (FCFS) scheduling can be stochastically and asymptotically improved upon by the {\it Nudge} scheduling algorithm in case of light-tailed job size distributions. Such improvements are feasible even when the jobs are partitioned into two types and the scheduler only has information about the type of incoming jobs (but not their size). In this paper we introduce Nudge*$(M)$ scheduling, where basically any incoming type-1 job is allowed to pass any type-2 job that is still waiting in the queue given that it arrived as one of the last $M$ jobs. We prove that Nudge*$(M)$ has an asymptotically optimal response time within a large family of  Nudge scheduling algorithms when job sizes are light-tailed. Simple explicit results for the asymptotic tail improvement ratio (ATIR) of Nudge*$(M)$ over FCFS are derived as well as explicit results for the optimal parameter $M$. An expression for the ATIR that only depends on the type-1 and type-2 mean job sizes and the fraction of type-1 jobs is presented in the heavy traffic setting. The paper further presents  a numerical method to compute the  response time distribution and mean response time of Nudge*$(M)$ scheduling provided that the job size distribution of both job types follows a phase-type distribution (by making use of the framework of Markov modulated fluid queues with jumps). 
}


\maketitle

\section{Introduction}

Scheduling is the process of arranging, controlling and optimizing the execution of (computing) tasks, called jobs, such that some objective function is minimized.  Examples of objective functions include the mean or maximum response time, quantiles of the response time distribution, the fraction of jobs meeting their deadline, the fairness of the scheduling algorithm, etc. The response time is defined as the time a job spends in the system until it leaves, accounting for both the time that the job is waiting and is processed. Examples of some popular scheduling algorithms include 
\begin{itemize}
    \item First-Come-First-Served (FCFS) which serves the jobs in their order of arrival, 
\item	Shortest-Remaining-Time-First (SRPT) which serves the job with the smallest remaining processing time that is present in the system, 
\item Least-Attained-Service (LAS) which serves the job which has received the least amount of processing thus far,
\item Earliest-Deadline-First (EDF) which serves the job with the smallest due date among all jobs that remain in the system. 
\end{itemize}
The choice of the scheduling algorithm depends on the objective function that one is trying to minimize. For instance, 
SRPT is known to minimize the mean response time in an M/G/1 queue \citep{schrage1}, FCFS minimizes the maximum response time 
of any finite sequence of jobs, EDF meets all the job deadlines whenever feasible \citep{Goldberg1977,liulayland3}, etc. 

This paper focuses on scheduling algorithms, called Nudge algorithms, that try to avoid long response times as much as possible instead of trying to minimize mean response times. These algorithms resemble FCFS, but occasionally serve some jobs in a different
order to reduce the probability of having long response times.  Such algorithms reflect common behavior that people do not mind waiting slightly longer from time to time if this is compensated by a much shorter wait on later occasions (when they are perhaps more time-constrained). For instance, consider a grocery store where many people regularly come to buy a large number of items, but occasionally buy just one or two items (because they forgot something during their previous visit or due to a promotion that is only valid on a specific day). In such cases, people are usually willing to let someone who is buying only one or two items pass them at the checkout counter. This small additional delay during their regular shopping visit is acceptable, as it may save them a more
substantial amount of time during one of
their own quick trips to purchase just a few items. In such a setting people are clearly not interested in minimizing mean
response times as this would lead to extremely long waits for people buying lots of items.
Thus while deviations from FCFS are often considered as a form of injustice \citep{Larson87}, there are settings where
such deviations are considered useful and common. The study of Nudge algorithms mathematically supports people's
intuition for the usefulness of such deviations.

In this paper we study a large class of Nudge algorithms with the objective to minimize the tail behavior of the response time. 
More formally, let $R$ be the random variable representing the response time and $P[R > t]$ the probability that
$R$ exceeds $t$. Then, our aim is to find the Nudge algorithm that minimizes this probability for large $t$.
One might argue that focusing on large $t$ does not provide much guarantees for smaller $t$ values. However, as shown in \citep{nudge,vanhoudt_nudge},
Nudge algorithms that perform well in reducing $P[R>t]$ for $t$ large often also do so for any $t > 0$.
More specifically we state that an algorithm $A$ asymptotically improves another algorithm $B$ if $P[R > t]$
of $A$ for $t$ large enough is below that of $B$. We state that $A$ stochastically improves $B$ if $P[R > t]$
of $A$ is less than that of $B$ for any $t$. In \citep{nudge,vanhoudt_nudge} it is shown that an asymptotic
improvement often yields a stochastic improvement as well. Similar findings are presented for
the Nudge*$(M)$ algorithm introduced in this paper.

For a long time it was unclear whether one can actually asymptotically improve upon FCFS. 
More specifically FCFS is known to be weakly tail optimal for so-called class-I job size distributions \citep{boxma07}.

A class-I job size distribution $X$ is a light-tailed distribution for which some mild technical conditions hold
such that the response time in an M/G/1 queue, where $X$ represents the service time, has exponential decay \citep{abate94}.
These distributions include all phase-type distributions \citep{latouche1} as well as any distribution with finite support.
Recall that a distribution is light-tailed if there exists an $\epsilon > 0$ such that $E[e^{\epsilon X}]$ is finite.
Thus for any class-I job size distribution $X$, we have 
\[P[R_{FCFS} > t] \sim c_{FCFS} e^{-\theta_Z t},\]
where $R_{FCFS}$ is the job response time in an M/G/1 queue with job size distribution $X$ under the FCFS 
scheduling discipline \cite[Section 5]{abate94}. The fact that FCFS is weakly tail optimal means that FCFS has
the highest possible decay rate $\theta_Z$ of all scheduling disciplines. In fact this decay rate is equal to
the decay rate of the workload in the queue (which is the same for any work conserving scheduling discipline).
Many other scheduling algorithms such as LAS and SRPT have a slower decay rate equal to that of the busy
period of the queue \citep{boxma07,nuyens08}. Thus, while FCFS was known to minimize the decay rate, it was unclear
whether the {\it prefactor} $c_{FCFS}$ could be further reduced.

In a recent paper \citep{nudge} it was shown that a lower prefactor can be achieved by the so-called {\it Nudge} 
scheduling algorithm (meaning FCFS is not strongly tail optimal). The Nudge algorithm in
\citep{nudge} operates such that the scheduler needs to know whether the size of an incoming job exceeds certain thresholds.
It was subsequently shown in \citep{vanhoudt_nudge} that similar results can be obtained in a much more relaxed setting where
jobs are partitioned into two types and the scheduler only needs to know the type of incoming jobs. This was
done by introducing the Nudge-$K$ scheduling algorithm. Nudge-$K$ operates as follows: 
{\it when a type-1 job arrives at time 
$t$ and the previous $k \leq K$ arrivals were type-2, then the incoming type-1 job is served before any of these
$k$ type-2 jobs that are still waiting in the queue at time $t$.} Note that under Nudge-$K$ a type-1 job can pass up to
$K$ type-2 jobs, but a type-2 job can be passed at most once. 

For this reason and to avoid confusion with other Nudge algorithms introduced in this paper, 
we will refer to the Nudge-$K$ algorithm as Nudge$^1(K)$. 
Assuming that the system is composed of just two job types may seem restricted at first. However, the main motivation
for this setting is that it implies that the scheduler only requires coarse information about the incoming jobs, which may be
easier to obtain than exact job sizes.

Let $R_{Nudge^1(K)}$ be the response time in an M/G/1 queue using the Nudge$^1(K)$ scheduling algorithm and
denote the Laplace transform of the job size distribution, the type-1 job size distribution
and the type-2 job size distribution
as $\tilde S(s)$, $\tilde S_1(s)$ and $\tilde S_2(s)$, respectively. It was shown in \citep{vanhoudt_nudge} that 
\[P[R_{Nudge^1(K)} > t] \sim \bar c_K e^{-\theta_Z t},\]
for some prefactor $\bar c_K$ and this prefactor is minimized over $K \geq 0$ by setting $K=\max(0,M_{opt})$ with
\begin{align}\label{eq:Kopt}
M_{opt} = \left\lfloor \left. \log \left( \frac{\tilde S_1(-\theta_Z) (\tilde S_2(-\theta_Z)-1)}
    {\tilde S_2(-\theta_Z) (\tilde S_1(-\theta_Z)-1)}\right) \middle/ \log(\tilde S(-\theta_Z)) \right.\right\rfloor.
\end{align}
When $M_{opt} \leq 0$, then $c_{FCFS} < \bar c_K$ for any $K >0$, while otherwise
$c_{FCFS} > \bar c_K$ for $K=M_{opt}$.

Similar to Nudge$^1(K)$, the optimal Nudge*$(M)$ algorithm introduced in this paper requires the knowledge
of the quantities $\tilde S_1(-\theta_Z)$, $\tilde S_2(-\theta_Z)$ and $\tilde S(-\theta_Z)$.  In case these quantities
are hard to estimate, we introduce a heavy traffic approximation that only requires knowledge of the arrival
rate $\lambda$, the mean type-1 and type-2 job sizes and the second moment of the overall job size distribution.

The objective of this paper is to study how small the prefactor can be made by a Nudge-like scheduling algorithm that 
only uses job arrival types to make scheduling decisions. For this purpose we introduce the Nudge*$(M)$ scheduling algorithm.
{\it When a type-1 job arrives at time $t$ under Nudge*$(M)$, it is served before any type-2 job still waiting in the queue
provided that this type-2 job was among the last $M$ arrivals (before time $t$). }
Rather surprisingly, we show that for the response time $R_{Nudge^*(M)}$ of Nudge*$(M)$ for any class-I job size distribution:
\[P[R_{Nudge^*(M)} > t] \sim c_{M} e^{-\theta_Z t},\]
and minimizing $c_M$ is done by setting $M=\max(0,M_{opt})$. This means that Nudge*$(M)$
has a lower prefactor $c_M$ than FCFS if and only if Nudge$^1(M)$ has a lower prefactor $\bar c_M$. Further the optimal parameter value
of $M$ and $K$ coincide. 

More importantly we also introduce a large family $\mathcal{F}$ of Nudge-like scheduling policies (see Section \ref{sec:model}) and
 show that Nudge*$(M)$ with $M=M_{opt}$ minimizes the prefactor among all the scheduling algorithms in this family $\mathcal{F}$,
meaning Nudge*$(M)$ with $M=M_{opt}$ is strongly tail optimal within $\mathcal{F}$.
Our results in the paper are presented using the asymptotic tail improvement ratio (ATIR) between two scheduling algorithms
$A_1$ and $A_2$ (as in \citep{nudge,vanhoudt_nudge}) which is defined by
\[ ATIR = 1 - c_{A_1}/c_{A_2},\]
where $c_{A_i}$ is the prefactor of scheduling algorithm $A_i$, for $i=1,2$. In our case $A_2$ will be FCFS and
$A_1$ is a Nudge scheduling algorithm belonging to $\mathcal{F}$. Notice that minimizing the prefactor corresponds to
maximizing the ATIR of $A_1$ with $A_2=FCFS$.

The Nudge*$(M)$ scheduling algorithm only uses the type of arriving jobs and the order in which
they arrive. One may wonder to what extent it can be further improved if more information is used, such as the
exact arrival times.
In a closely related paper \citep{boostZiv} that was written concurrently to this paper, the authors introduce the $\gamma$-Boost scheduling algorithm. This algorithm minimizes the prefactor for class-I job size distributions in an M/G/1 queue among all scheduling policies in case the job size $a$ of each individual job 
as well as the arrival time of each job is known. The idea is to boost the arrival time of a size $a$ job by $b(a)=\log( 1/(1-e^{-\theta_Z a}) )/\theta_Z$, thus a size $a_2$ job
that arrives at time $t_2$ is served before a size $a_1$ job that arrived at time $t_1 < t_2$ if
$t_2-b(a_2) < t_1-b(a_1)$ provided that the size $a_1$ job is still waiting in the queue at time $t_2$. 
Moreover the authors also propose a $\gamma$-Boost algorithm in case the jobs are partitioned 
into several types and the
scheduler only has information about the job types and arrival time, but not the individual sizes. 
In this setting the boost of a job depends on its type only.
% Contrary to the Nudge-like algorithms considered in this paper, 
%$\gamma$-Boost uses arrival time information to schedule jobs in addition to the job types, 
%which is not required for Nudge*$(M)$. 
The authors of \citep{boostZiv} prove that $\gamma$-Boost achieves a lower prefactor than the Nudge*$(M)$
algorithm by exploiting the additional arrival time information.
In this paper we prove that the prefactor of $\gamma$-Boost and Nudge*$(M)$ coincide in the heavy traffic limit. This indicates that the gain offered by the additional arrival time information vanishes as the load tends to one.

In fact \citep{boostZiv}  was not the first paper to consider scheduling jobs based on boosted arrival times.
Such an algorithm was previously proposed and analyzed in a discrete-time setting with two job types
in \citep{timelimited}. This algorithm relies on a discrete parameter $N$ such that when a type-1 job arrives 
it is allowed to pass any type-$2$ job that is still waiting and that arrived in the last $N$ time slots. 
The analysis was however focused on deriving generating functions for the response time of type-$1$ and type-$2$ jobs separately
for a given $N$.

Other discrete-time systems with limited job passing were studied in  \citep{Reserv} and \citep{ReservTail}.
The first of these papers considered a queue with two job types where all the jobs have size one and is in fact equivalent to the 
Nudge$^1(K)$ algorithm with $K=\infty$, meaning a type-1 job may pass any type-$2$ job that is still in
the waiting room and that arrived since the last type-$1$ arrival. In the second paper a response time tail analysis was presented for
a more general class of scheduling algorithms (that are not necessarily equivalent to some Nudge algorithm), but again
in the assumption that all jobs have size one. Although \citep{ReservTail} considers more than two job types, the analysis 
is equivalent to a system with just two job types due to the nature of these algorithms and the
fact that all jobs have the same size. Note that when all the jobs have the
same size $\gamma$-Boost reduces to FCFS and a stochastic improvement cannot be achieved by changing the order of jobs 
(as the mean response time remains the same).

It is worth noting that the expression for $M_{opt}$ in \eqref{eq:Kopt} can be rewritten as
\begin{align}\label{eq:Kopt2}
M_{opt} = \left\lfloor \left.  \log \left( \frac{\tilde S_1(-\theta_Z)}
    {\tilde S_1(-\theta_Z)-1}\right)\middle/ \log(\tilde S(-\theta_Z)) -\log \left( \frac{\tilde S_2(-\theta_Z)}
    {\tilde S_2(-\theta_Z)-1}\right)\middle/ \log(\tilde S(-\theta_Z))  \right.\right\rfloor.
\end{align}
Thus if we define $b_i = \log (\tilde S_i(-\theta_Z)/ (\tilde S_i(-\theta_Z)-1) )/ \log \tilde S(-\theta_Z)$,
then a type-$1$ job passes a type-$2$ job under Nudge*$(M)$ with $M=M_{opt}$ if and only if the number of arrivals
between the type-$1$ and the type-$2$ job is less than $\lfloor b_1 - b_2 \rfloor$.
We can therefore think of the Nudge*$(M)$ algorithm as an algorithm that schedules jobs based on a boosted
arrival order, whereas $\gamma$-Boost schedules jobs using a boosted arrival time.

The main contributions of the paper are the following: 
\begin{enumerate}
    \item We introduce the Nudge*$(M)$ scheduling algorithm and derive simple explicit results for the ATIR 
    of Nudge*$(M)$
    as well as for the parameter $M$ that maximizes the ATIR of Nudge*$(M)$. These explicit results are
    expressed in terms
    of the Laplace transforms $\tilde S(s)$, $\tilde S_1(s)$ and $\tilde S_2(s)$ evaluated in $s=-\theta_Z$. 
    \item We present explicit results for the ATIR of Nudge*$(M)$, with $M$ optimized, in heavy traffic (as $\lambda$ tends to one)
    that depend only on $p$, $E[X_1]$, $E[X_2]$ and show that the ATIR coincides with that
    of the $\gamma$-Boost algorithm in \citep{boostZiv} in the limit.
    \item We prove that Nudge*$(M)$, with $M=M_{opt}$, maximizes the ATIR in  a large family of
    Nudge scheduling algorithms $\mathcal{F}$. 
    \item We present a numerical method to compute the mean response time of Nudge*$(M)$ as well as
    the type-1 and type-2 job response time distribution. Numerical results show that Nudge*$(M)$ also
    yields significant gains for the mean response time.
\end{enumerate}
The first three contributions are established for any class-I job size distribution.
An extended 3-page abstract that summarizes the first three contributions was presented at the ACM Sigmetrics MAMA workshop
on June 14th 2024 \citep{charletMAMA}. 
The fourth contribution is made
under the assumption that job sizes follow a phase-type distribution. This numerical method is very different from the methods used in  \citep{nudge} and  \citep{vanhoudt_nudge}, because it is no longer possible to create a simple link between the Nudge*$(M)$ queue and the FCFS queue. This forced us to use a radically different approach.

The paper is composed of nine sections followed by references and an E-companion, the six sections of which are labeled as
EC.1 to EC.6. Section \ref{sec:model} presents the model  as well as 
the family of scheduling algorithms under consideration. Explicit results for the ATIR and optimal
$M$ for Nudge*$(M)$ are derived in Section \ref{sec:ATIR}, where the heavy traffic setting is studied in
Section \ref{sec:heavy}. The optimality of Nudge*$(M)$ in $\mathcal{F}$ with respect to the ATIR is established in Section \ref{sec:opt}. In Section \ref{sec:split} we look at the problem of how to best split arriving jobs into two types based on their size in order to minimize the tail of the response time using Nudge*$(M)$, while Section \ref{sec:multi} illustrates that
Nudge$^*(M)$ can also reduce tails in a multi-server setting. 
A numerical method to compute the response time distribution of a type-2 job under Nudge*$(M)$ is developed in Section \ref{sec:t2}, while the same is done for a type-1 job in Section \ref{app:t1}. Conclusions are drawn and future work is discussed in Section \ref{sec:conc}. In Section \ref{app:mean} the mean response time of Nudge*$(M)$ is analyzed.
Some of the proofs in the main body of the paper are presented in the first four sections of the E-companion.

\section{Model and Algorithms}\label{sec:model}

We consider a queueing system with two types of jobs. Arriving jobs are either type-1 with probability $p$ or type-2 with probability $1-p$, and consecutive types are independent. Jobs arrive following a Poisson process with parameter $\lambda$.

The system may actually be composed of many more types of jobs, but our assumption is that from the perspective of the
scheduler a job is either type-1 or type-2. We show that even with such coarse information, Nudge algorithms achieve
significant performance gains. If more job size information is available to the scheduler, 
then a further reduction of the prefactor can be achieved
by developing Nudge algorithms for systems with $d > 2$ types.
Given the remark in the introduction that the Nudge*$(M)$ algorithm can be regarded as an algorithm that schedules
jobs based on a boosted arrival order (see \eqref{eq:Kopt2}), a natural generalization to a setting with
more than $2$ types would be to give each job a boost in its arrival index based on its type and to schedule the
jobs based on the boosted arrival order.

Denote $\tilde{S_i}(s)$ as the Laplace transform of a type-$i$ job size, for $i=1,2$ and let
$\tilde{S}(s) = p \tilde{S_1}(s) + (1-p) \tilde{S_2}(s)$ be the Laplace transform of a random job size. Let $E[X_i]$ be the mean job size of a type-$i$ job, for $i=1,2$. Without loss of generality, assume that $p E[X_1] + (1-p) E[X_2] = 1$ so that the load of the system is $\lambda$.  
In order to optimize the parameter $M$ of the
Nudge*$(M)$ algorithm, the scheduler needs to know $\tilde S_1(-\theta_Z), \tilde S_2(-\theta_Z)$ and $\tilde S(-\theta_Z)$,
as these numbers are used to compute $M_{opt}$ (see \eqref{eq:Kopt}).
When these numbers are hard to estimate in a real system, it is also possible to rely on a heavy traffic approximation
$M_{heavy}$ for $M_{opt}$ that only requires estimates for the mean type-1 and type-2 job sizes $E[X_1]$ and $E[X_2]$,
the second moment of the overall job size $E[X^2]$ and the arrival rate $\lambda$ (see Section \ref{sec:heavy}).

We further note that we assume throughout the paper that the two job types are labeled such that $\tilde S_1(-\theta_Z) < \tilde S_2(-\theta_Z)$. In case $\tilde S_1(-\theta_Z) \geq \tilde S_2(-\theta_Z)$ one readily finds by \eqref{eq:Kopt2} that $M_{opt}\leq 0$, which means that Nudge algorithms do not provide an asymptotic improvement over FCFS.  
When working with the heavy traffic approximation
the labeling is done such that $E[X_1] < E[X_2]$.

In Sections \ref{sec:ATIR} to
\ref{sec:split} we demand that the job size distribution is a class-I distribution.
Denote $\theta_i$ as the decay rate of the type-$i$ job size distribution, for $i=1,2$.
In Section \ref{sec:t2} we make a slightly stronger assumption by assuming that the service time of a type-$i$ job follows a phase-type distribution $X_i \sim PH(\alpha_i, S_i)$ with $n_i$ phases.  Define $\alpha = (p \alpha_1, (1-p) \alpha_2)$,
$$ S = \begin{bmatrix} S_1 & 0 \\ 0 & S_2 \end{bmatrix},$$ $s_i^* = (-S_i) \e$, and $s^* = (-S) \e$ for further use.
A random job's service time follows a phase-type distribution with parameters $(\alpha, S)$: $X = p X_1 + (1-p) X_2$ and $E[X] = \alpha (-S)^{-1} \e = 1$. We assume that $\alpha \e =1$, meaning all jobs have a nonzero size.  
In this case we have  $\tilde{S_i}(s) = \alpha_i (sI - S_i)^{-1} s_i^*$ and $\tilde{S}(s) = \alpha (sI - S)^{-1} s^*$. It is well known that any general positive-valued distribution can be approximated arbitrary close with a PH distribution  \citep{latouche1}. Further, various fitting algorithms
and tools are available online 
for phase-type distributions (e.g., \citep{feldman98,panchenko1,Kriege2014}).
Whenever we use a hyperexponential (HE) distribution for the type-$i$ jobs (with 2 phases) we match
the mean $E[X_i]=q/\mu_1+(1-q)/\mu_2$, the squared coefficient of variation (SCV) and the parameter $f = (q/\mu_1)/E[X_i]$, 
when the HE distribution is characterized by $q, \mu_1$ and $\mu_2$.

The original Nudge algorithm introduced in \citep{nudge} made use of $4$ types of jobs: small, medium, large and huge jobs.
These jobs were classified based on their size. A small job was allowed to pass a large job that arrived just before
the small job, provided that the large job was still waiting. The division in four types was needed to formally prove that
for any type-I job size distribution $X$, one can find three thresholds such that if the jobs are split based on these thresholds
the Nudge algorithm stochastically improves upon FCFS. The authors however provided numerical evidence that this was also
the case with just two job types, that is, without the medium and huge jobs. In Section \ref{sec:split} we consider the
problem of finding the optimal threshold to split the jobs in two types based on their size.
In \citep{vanhoudt_nudge} the set of Nudge
algorithms was generalized by allowing type-1 jobs to pass up to $K$ type-$2$ jobs, while type-$2$ jobs can
still be passed at most once. A natural further generalization is to allow that type-$2$ jobs can be passed by
multiple type-$1$ jobs as well. In fact, this is what the Nudge$^{12}(K,L)$ algorithm introduced below does. We initially focused
our attention on this class of Nudge algorithms. However optimizing the tail of the response time in terms of the parameters
$K$ and $L$ turns out to be hard as there is no simple closed form formula for the optimal $K$ and $L$. This led us to
the introduction of the Nudge*$(M)$ algorithm, which is much easier to optimize, and contrary to our initial expectations
turned out to be superior to Nudge$^{12}(K,L)$. This in turn motivated us to prove optimality of Nudge*$(M)$
in a much wider family $\mathcal{F}$ of Nudge algorithms that is defined next.

Algorithms belonging to $\mathcal{F}$ allow type-1 jobs to occasionally pass one or more type-2 jobs. More precisely, for any algorithm $A \in \mathcal{F}$ there exists an $M$ such that 
when a type-1 job arrives at time $t$ it will look at the types of the last $M$ arrivals before time $t$. 
Algorithm $A$ will
be characterized by a function $n$, which takes the types of the last $M$ arrivals as input, that will specify how
many type-$2$ jobs the type-$1$ job is allowed to pass (provided they are still waiting in the queue). 
Some restrictions on the function
$n$ need to be imposed in order to come up with a feasible algorithm. The main reason for this is that whenever a type-1 job
passes some type-2 job, any intermediate type-1 job should also pass this type-2 job.
We now proceed with a formal definition of the family $\mathcal{F}$.

We first define the family $\mathcal{F}_M$ of Nudge policies, $\mathcal{F}$ is then obtained
by taking the union over all $M$.  Let $t$ be the function 
that counts the number of twos in a string of any length consisting of ones and twos, e.g., $t(12122)=3$. 
A Nudge scheduling algorithm belonging to $\mathcal{F}_M$ is characterized by a function $n$ from $\{1,2\}^M$ to $\{0,\ldots,M\}$ that
obeys the following two conditions:
\begin{itemize}
    \item[(C1)] $n(s) \leq t(s)$, 
    \item[(C2)] $n(s_0 s_1 \ldots s_{M-1}) \leq n(s)+1(s_0=2)$
\end{itemize}
 for all $s= s_1\ldots s_M \in \{1,2\}^M$, where $1(A)=1$ if $A$ is true and $1(A)=0$ otherwise. 
As stated before, whenever a type-1 job arrives it looks at the types of the last $M$ arrivals. Assume these $M$ types
are characterized by the string $s$, then the type-1 job passes the $n(s)$ most recent type-2 arrivals
if they are still waiting in the queue. For instance, if $n(s)=3$, but there are only two type-2 jobs
waiting in the queue, then the type-1 job passes only these two type-2 jobs.
Note that the condition (C2) on $n(s)$ guarantees that if a type-1
job may pass a type-2 job, all intermediate arrivals of type-1 may also pass this type-2 job.
We clearly have that $\mathcal{F}_M \subset \mathcal{F}_{M+1}$, for any $M$, as for any $n(s)$ in $\mathcal{F}_M$, 
we can define $n'(s)$ in $\mathcal{F}_{M+1}$
such that $n'(s_1 \ldots s_M 1) = n'(s_1 \ldots s_M 2) = n(s_1 \ldots s_M)$. 
The family $\mathcal{F}$ is defined as $\mathcal{F} = \bigcup_{M \geq0} \mathcal{F}_M$ and contains a large variety of
scheduling policies that use the last $M$ job arrival types for some $M$ to make scheduling decisions.

Examples of Nudge policies that are part of the family $\mathcal{F}$ are the following:
\begin{enumerate}
    \item {\bf Nudge*$(M)$ $\in \mathcal{F}_M$:} Under this scheduling algorithm a type-1 job passes any type-2 job still waiting in the
    queue that arrived among the last $M$ arrivals.  We stress that under Nudge*$(M)$ it is the
    order of the last $M$ arrivals that matters and not the order of the last $M$ jobs in the queue. These two orders
    are obviously not necessarily the same if one of the last $M$ arrivals either passed another job or was passed.  
    It is defined by letting $n(s) = t(s)$ for all strings $s$.
    Under this scheduling algorithm a type-1 job can pass at most $M$ type-2
    jobs and a type-2 job is passed at most $M$ times.
    \item {\bf Nudge$^1(K)$ $\in \mathcal{F}_K$:}  Under this scheduling algorithm, 
    studied in \citep{vanhoudt_nudge}, a type-1 job passes at most $K$ type-2 jobs and a 
    type-2 job is passed at most once. It is defined by setting $n(s)$ as the number of leading twos in $s$ before encountering a one in $s$, for example $n(22212) = 3$. This clearly satisfies (C1), and (C2) is also satisfied as $n(s_0 s_1 \ldots s_{K-1}) = n(s)+1$ if $s_0 = 2$ (except if $n(s_0 s_1 \ldots s_{K-1})=n(s)=K$), or 
    $n(s_0 s_1 \ldots s_{M-1})=0 \leq n(s)$ if $s_0 = 1$.
    \item {\bf Nudge$^2(L)$ $\in \mathcal{F}_L$:} Under this scheduling algorithm a type-1 job can pass at most one type-2
    job and a type-2 job is passed at most $L$ times. As such the scheduling algorithm can be regarded at the {\it dual} of the 
    Nudge$^1(K)$ scheduling algorithm. In this case $n(s)=\min(t(s),1)$. It is not hard to show that the ATIR of Nudge$^2(L)$ is maximized by setting $L=\max(0,M_{opt})$ (in fact, this is a consequence of Lemma \ref{th:ns} in this paper). Hence, the optimal parameter values of Nudge*$(M)$, Nudge$^1(K)$ and Nudge$^2(L)$
    all coincide.
    \item {\bf Nudge*$^{1}(K,M)$ $\in \mathcal{F}_M$:} This scheduling algorithm is the same as Nudge*$(M)$, but we limit the number
    of type-2 jobs that a type-1 job can pass to $K \leq M$.
    It is defined by the function $n(s) = \min (t(s), K)$. If $K = M$, the scheduling algorithm coincides with Nudge*$(M)$ and setting
    $K=1$ yields Nudge$^2(M)$.  
    \item {\bf Nudge*$^2(M,L)$ $\in \mathcal{F}_M$:} This scheduling algorithm is the same as Nudge*$(M)$, but we limit the number
    of times a type-2 jobs can be passed to $L \leq M$.
    It is defined by the function $n(s)$ where $n(s)$ counts the number of twos before encountering the
    $L$-th one in $s$. If $L = M$, the scheduling algorithm coincides with Nudge*$(M)$, while for $L=1$ it coincides with Nudge$^1(M)$.  
    \item {\bf Nudge$^{12}(K,L)$ $\in \mathcal{F}_{K+L-1}$}. Under this scheduling algorithm a type-1 job can pass up to $K$ type-2
    jobs, while a type-2 job is passed at most $L$ times. It can be regarded as a combination of Nudge*$^1(K,M)$ and Nudge*$^2(M,L)$. 
    For this scheduling algorithm, $n(s)$ is defined as follows: count from left to right the number of ones and twos in $s$, stopping either when $K$ twos have been counted or when $L$ ones have been counted. $n(s)$ is then the number of twos counted, which corresponds to at most $K$ type-2 jobs that have been swapped fewer than $L$ times. As an example, with $K = 3$ and $L = 2$ we have $n(21\underline{1}2) = 1$, $n(22\underline{2}2) = 3$, and $n(1\underline{1}22) = 0$, where the underlined type is the last one that is counted. Clearly (C1) holds, and (C2) is also satisfied
    as explained next. If $s_0 = 1$, then an extra type-1 job is counted before continuing as in $s$. This means the counting is either stopped at the same time in $s$ if $K$ twos were seen, or sooner than in $s$ if $L$ ones were seen, $L-1$ of which are in $s$. If $s_0 = 2$, an extra type-2 job is counted before continuing as in $s$. Either the counting stops when seeing $K$ twos, $K-1$ of which were also counted in $s$, or the counting stops when counting $L$ ones and $i < K$ twos, which is an extra two compared to the result $n(s)$ of $s$.
   \end{enumerate}

Figure \ref{fig:illustration} gives an example of how Nudge$^1(3)$ and Nudge$^*(3)$ schedule jobs. In case of FCFS, jobs are served in arrival order, meaning job A is served first and job K is served last. Jobs with a hat are type-1 jobs, while other jobs are type-2. We assume that this sequence of arrivals occurred while the server was still working on some prior jobs (as jobs are only passed while
waiting in the queue). 
For Nudge$^1(3)$, when job C arrives, it is able to pass jobs A and B. When job E arrives, it sees jobs D, B, and A at the back of the queue, but it can only pass job D: jobs B and A have already been passed before and cannot be passed multiple times in this policy. This restriction is lifted for Nudge*$(3)$: job C still passes jobs A and B, but job E can now also pass B. While job A is the third job at the back of the queue when job E arrives, it is {\it not} one of the previous 3 arrivals, and so it is {\it not} passed by job E. Job J can pass three jobs for both policies, while job K is not able to pass any jobs when using Nudge$^1(3)$ as jobs I, H, and G have already been passed by J.

\begin{figure*}[t!]
\centering
    \begin{tikzpicture}

\draw[very thick,->] (-11, -1.4000000000000001) -- (1, -1.4000000000000001);
\draw (0, 0) circle (0.3) node {};
\draw (0, 0) node[anchor=center]{A};
\draw[very thick] (0, -0.3) -- (0, -1.0);
\draw[very thick] (-0.3, -0.7) -- (0, -0.4) -- (0.3, -0.7);
\draw[very thick] (-0.2, -1.3) -- (0, -1.0) -- (0.2, -1.3);
\draw (-1, 0) circle (0.3) node {};
\draw (-1, 0) node[anchor=center]{B};
\draw[very thick] (-1, -0.3) -- (-1, -1.0);
\draw[very thick] (-1.3, -0.7) -- (-1, -0.4) -- (-0.7, -0.7);
\draw[very thick] (-1.2, -1.3) -- (-1, -1.0) -- (-0.8, -1.3);
\draw (-2, 0) circle (0.3) node {};
\draw (-2, 0) node[anchor=center]{C};
\draw[very thick] (-2, -0.3) -- (-2, -1.0);
\draw[very thick] (-2.3, -0.7) -- (-2, -0.4) -- (-1.7, -0.7);
\draw[very thick] (-2.2, -1.3) -- (-2, -1.0) -- (-1.8, -1.3);
\draw (-2.4, 0.3) -- (-1.6, 0.3) -- (-2, 0.6) -- cycle;
\draw[->,thick] (-2, 0.7) to[out=40,in=140] (0.5, 0.7);
\draw[->,dotted,thick] (-2, -1.6) to[out=-40,in=-140] (0.5, -1.6);
\draw (-3, 0) circle (0.3) node {};
\draw (-3, 0) node[anchor=center]{D};
\draw[very thick] (-3, -0.3) -- (-3, -1.0);
\draw[very thick] (-3.3, -0.7) -- (-3, -0.4) -- (-2.7, -0.7);
\draw[very thick] (-3.2, -1.3) -- (-3, -1.0) -- (-2.8, -1.3);
\draw (-4, 0) circle (0.3) node {};
\draw (-4, 0) node[anchor=center]{E};
\draw[very thick] (-4, -0.3) -- (-4, -1.0);
\draw[very thick] (-4.3, -0.7) -- (-4, -0.4) -- (-3.7, -0.7);
\draw[very thick] (-4.2, -1.3) -- (-4, -1.0) -- (-3.8, -1.3);
\draw (-4.4, 0.3) -- (-3.6, 0.3) -- (-4, 0.6) -- cycle;
\draw[->,thick] (-4, 0.7) to[out=40,in=140] (-0.5, 0.7);
\draw[->,dotted,thick] (-4, -1.6) to[out=-40,in=-140] (-2.5, -1.6);
\draw (-5, 0) circle (0.3) node {};
\draw (-5, 0) node[anchor=center]{F};
\draw[very thick] (-5, -0.3) -- (-5, -1.0);
\draw[very thick] (-5.3, -0.7) -- (-5, -0.4) -- (-4.7, -0.7);
\draw[very thick] (-5.2, -1.3) -- (-5, -1.0) -- (-4.8, -1.3);
\draw (-6, 0) circle (0.3) node {};
\draw (-6, 0) node[anchor=center]{G};
\draw[very thick] (-6, -0.3) -- (-6, -1.0);
\draw[very thick] (-6.3, -0.7) -- (-6, -0.4) -- (-5.7, -0.7);
\draw[very thick] (-6.2, -1.3) -- (-6, -1.0) -- (-5.8, -1.3);
\draw (-7, 0) circle (0.3) node {};
\draw (-7, 0) node[anchor=center]{H};
\draw[very thick] (-7, -0.3) -- (-7, -1.0);
\draw[very thick] (-7.3, -0.7) -- (-7, -0.4) -- (-6.7, -0.7);
\draw[very thick] (-7.2, -1.3) -- (-7, -1.0) -- (-6.8, -1.3);
\draw (-8, 0) circle (0.3) node {};
\draw (-8, 0) node[anchor=center]{I};
\draw[very thick] (-8, -0.3) -- (-8, -1.0);
\draw[very thick] (-8.3, -0.7) -- (-8, -0.4) -- (-7.7, -0.7);
\draw[very thick] (-8.2, -1.3) -- (-8, -1.0) -- (-7.8, -1.3);
\draw (-9, 0) circle (0.3) node {};
\draw (-9, 0) node[anchor=center]{J};
\draw[very thick] (-9, -0.3) -- (-9, -1.0);
\draw[very thick] (-9.3, -0.7) -- (-9, -0.4) -- (-8.7, -0.7);
\draw[very thick] (-9.2, -1.3) -- (-9, -1.0) -- (-8.8, -1.3);
\draw (-9.4, 0.3) -- (-8.6, 0.3) -- (-9, 0.6) -- cycle;
\draw[->,thick] (-9, 0.7) to[out=40,in=140] (-5.5, 0.7);
\draw[->,dotted,thick] (-9, -1.6) to[out=-40,in=-140] (-5.5, -1.6);
\draw (-10, 0) circle (0.3) node {};
\draw (-10, 0) node[anchor=center]{K};
\draw[very thick] (-10, -0.3) -- (-10, -1.0);
\draw[very thick] (-10.3, -0.7) -- (-10, -0.4) -- (-9.7, -0.7);
\draw[very thick] (-10.2, -1.3) -- (-10, -1.0) -- (-9.8, -1.3);
\draw (-10.4, 0.3) -- (-9.6, 0.3) -- (-10, 0.6) -- cycle;
\draw[->,thick] (-10, 0.7) to[out=40,in=140] (-6.5, 0.7);

    \end{tikzpicture}
    \caption{Example of how an arrival sequence is scheduled by Nudge$^1(3)$ and Nudge*$(3)$, compared to FCFS. Arrival A is the first while arrival K is the most recent, and jobs with a hat are type-1. The top arrows indicate how many jobs a type-1 job passes when using Nudge*$(3)$, while the dotted arrows below show how many jobs are passed when using Nudge$^1(3)$.}
    \label{fig:illustration}
\end{figure*}

In the next three sections we present results on the ATIR. We will often make use of the final value theorem and the 
properties of the Laplace transform in the following manner. Assume $Y=Y_1+Y_2$ with
$Y_1$ and $Y_2$ independent, where both $Y$ and $Y_1$
decay exponentially at rate $\theta_Z$, meaning $P[Y > t] \sim c_Y e^{-\theta_Z t}$ and
$P[Y_1 > t] \sim c_{Y_1} e^{-\theta_Z t}$ for some $c_Y$ and $c_{Y_1}$. Then,
\begin{align*}
    c_Y &= \lim_{t \rightarrow \infty} e^{\theta_Z t} P[Y > t] =
 \lim_{s \rightarrow 0} \frac{1}{\theta_Z} s\tilde Y(s-\theta_Z) = \lim_{s \rightarrow 0} \frac{1}{\theta_Z} s\tilde Y_1(s-\theta_Z)\tilde Y_2(s-\theta_Z) = 
c_{Y_1} \tilde Y_2(-\theta_Z).
\end{align*}
For example, under FCFS the waiting time distribution is equal to the workload distribution $Z$, and the response time is equal to the sum of the waiting time and the job size distribution. As $R = W + X = Z + X$, this implies in the previously described manner that $c_{FCFS} = c_{Z} \tS{}$.

\section{The Asymptotic Tail Improvement Ratio of Nudge*$(M)$}\label{sec:ATIR}

In this section we derive a simple closed form expression for the ATIR of Nudge*$(M)$ which we denote as
$ATIR^*(M)$. We show that this is a concave function of $M$ that reaches a maximum when $M=M_{opt}$, with
$M_{opt}$ given by \eqref{eq:Kopt}. 

As any job can either pass up to $M$ jobs or be passed by up to $M$ jobs
under Nudge*$(M)$, it is not hard to
prove that %the decay rate of
both the type-1 and type-2 response time distribution also
decay exponentially with rate $\theta_Z$ for any class-I job size distribution $X$. % $S$
To obtain an expression for $ATIR^*(M)$, we start by deriving expressions for the prefactors of $c_{W^{(i)}}(M)$ of the waiting time 
$W_M^{(i)}$ of a type-$i$ job under Nudge*$(M)$ scheduling for $i=1,2$.

\begin{lemma}\label{th:CW1}
    Let $c_{W^{(1)}}(M) = \lim_{t \rightarrow \infty} e^{\theta_Z t} P[W_M^{(1)} > t]$, then
\begin{align}
    c_{W^{(1)}}(M) &= c_Z (w_1+w)^M,\end{align}
with $w_1 = p\tilde S_1(-\theta_Z) / \tilde S(-\theta_Z)$ and
$w=(1-p)/\tilde S(-\theta_Z)$, which decreases in $M$ as $w_1+w < 1$.
\end{lemma}

\begin{proof}{Proof:}
 We provide a sketch of the proof, the formal proof is deferred to Section \ref{app:CW1}. 
When assessing the prefactor of the waiting time of a type-$1$ job, one first shows that we may assume
that upon arrival there are at least $M$ jobs waiting. The prefactor is then computed in three steps
using the final value theorem.
We first take $c_Z$ as the prefactor of the workload present in the queue upon arrival. Secondly we divide
this prefactor by $\tilde S(-\theta)^M$ to obtain the workload without the last $M$ jobs. However the type-$1$
jobs that are part of these $M$ jobs are not passed, so we need to add these back. This adds a factor
$p\tilde S_1(-\theta_Z) + (1-p)$ for each of these $M$ jobs as a random job is a type-$1$ job with
probability $p$. \Halmos

\end{proof}

\begin{lemma}\label{th:CW2}
    Let $c_{W^{(2)}}(M) = \lim_{t \rightarrow \infty} e^{\theta_Z t} P[W_M^{(2)} > t]$, then
\begin{align}
    c_{W^{(2)}}(M)= c_Z (w_1+w)^M
    \tilde S(-\theta_Z)^M,\end{align}
with $w_1 = p\tilde S_1(-\theta_Z) / \tilde S(-\theta_Z)$ and
$w=(1-p)/\tilde S(-\theta_Z)$, which increases in $M$ as $(w_1+w)\tilde S(-\theta_Z) = (1-p)+p\tilde S_1(-\theta_Z) > 1$.
\end{lemma}
\begin{proof}{Proof:}

We present a similar sketch as for the previous Lemma and refer to Section \ref{app:CW2} for a detailed proof.
The proof starts by showing that we may assume that there will be $M$ arrivals while the type-$2$ job is waiting
when computing the prefactor. Next we take the prefactor $c_Z$ of the workload and for each of these $M$ future arrivals
we need to add a factor $p \tilde S_1(-\theta_Z)+(1-p)$ as any type-$1$ job among the next $M$ passes the type-$2$ job.
\Halmos
\end{proof}

The next theorem presents the main result of this section and given an explicit expression for the ATIR of Nudge*$(M)$ over FCFS, 
which we denoted as $ATIR^*(M)$.

\begin{theorem}
For Nudge*$(M)$ we have
\begin{align}
ATIR^*(M) = 1-w_1 (w_1+w)^M -(1-w_1)(w_1+w)^M\tilde S(-\theta_Z)^M, 
\end{align}  
with $w_1 = p\tilde S_1(-\theta_Z) / \tilde S(-\theta_Z)$ and
$w=(1-p)/\tilde S(-\theta_Z)$.
Further, $ATIR^*(M)$ is concave and achieves a unique maximum in  $M = M_{opt}$
defined by \eqref{eq:Kopt}.
\end{theorem}

\begin{proof}{Proof:}
Given the expressions for $c_{W^{(1)}}(M)$ and $c_{W^{(2)}}(M)$ we have
\begin{align}
ATIR^*(M) &= 1 - p \frac{c_{W^{(1)}}(M) \tilde{S_1}(-\theta_Z)}{c_Z \tilde{S}(-\theta_Z)} - (1-p) \frac{c_{W^{(2)}}(M) \tilde{S_2}(-\theta_Z)}{c_Z \tilde{S}(-\theta_Z)} \\
&= 1-w_1 (w_1+w)^M -(1-w_1)(w_1+w)^M\tilde S(-\theta_Z)^M
\end{align}
This implies that the derivative with respect to $M$ is given by
\[
(w+w_1)^M \left( \tilde S(-\theta_Z)^M (w_1-1) \log((w_1+w)\tilde S(-\theta_Z)) - w_1 \log(w_1+w)  \right).
\]
and equals zero when
\[
\tilde S(-\theta_Z)^M 
= \frac{- w_1 \log(w_1+w)}{(1-w_1) \log((w_1+w)\tilde S(-\theta_Z))}.
\]
The second derivative of $ATIR^*(M)$ with respect to $M$ can be written as
\[ (w_1+w)^M \left( \tilde S(-\theta_Z)^M (w_1-1) \log((w_1+w)\tilde S(-\theta_Z))^2 
- w_1 \log(w_1+w)^2 \right), \]
which is negative as $0 < w_1 < 1$, proving that $ATIR^*(M)$ is concave in $M$. This proves that
the $ATIR^*(M)$ has a unique maximum.

Define $\Delta ATIR^*(M) = ATIR^*(M+1)-ATIR^*(M)$, then we have
\begin{align}
\Delta &ATIR^*(M) = (w_1+w)^M \left(w_1 w (\tilde S_2(-\theta_Z) - 1) - (1-w_1) \tilde S(-\theta_Z)^M p 
(\tilde S_1(-\theta_Z)-1) \right).
\end{align}
The maximum of $ATIR^*(M)$ is located in the ceil of the value of $M$ for which $\Delta ATIR^*(M) = 0$. 
Rewriting this equality yields \eqref{eq:Kopt}.
\Halmos \end{proof}
%It is easy to see that $M_{opt}< 0$ whenever $\tilde S_1(-\theta_Z) > \tilde S_2(-\theta_Z)$. In such case one can 
%reverse the role of the type-1 and type-2 jobs. The optimal $M$ for this reversed system is simply given by
%$-(M_{opt}+1)$.

In  Section \ref{sec:opt} we prove that Nudge*$(M)$ with $M=M_{opt}$ maximizes the ATIR in $\mathcal{F}$. To illustrate that
the gains over Nudge$^1(K)$,  studied in \citep{vanhoudt_nudge}, are often substantial we look at the ATIR of
Nudge*$(M)$ over Nudge$^1(K)$ in Figure \ref{fig:ATIRKcont} with $K=M=M_{opt}$ when both job sizes are exponential.
The ATIR of Nudge*$(M)$ is clearly much higher, especially under high loads $\lambda$. We further note that even higher
ATIR values can often be achieved for Nudge*$(M)$, for instance by making the type-2 job sizes more variable. Figure \ref{fig:ATIRMcont} shows the ATIR of $\gamma$-Boost over Nudge*$(M)$. While $\gamma$-Boost further improves
the ATIR over Nudge*$(M)$, the gains are less pronounced compared to the gains of Nudge*$(M)$ over Nudge$^1(K)$ (especially at high loads). Further, in the next section we show that the ATIR over FCFS of Nudge*$(M)$ and of $\gamma$-Boost converge to the same quantity when $\lambda$ tends to one.

\begin{figure*}[t!]
\begin{subfigure}{.48\textwidth}
  \centering
  \includegraphics[width=1\linewidth]{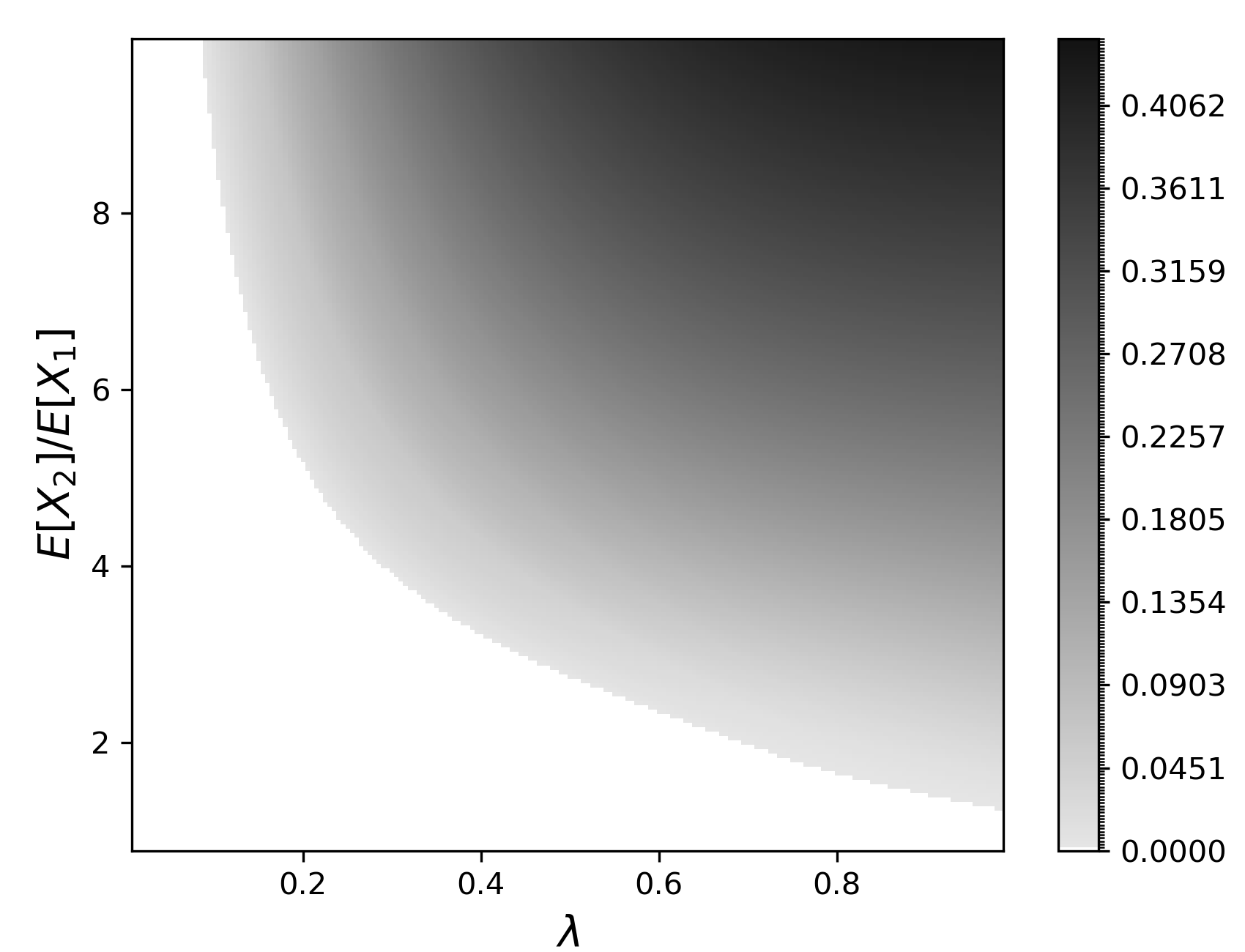}
	\caption{Plot of ATIR for Nudge*$(M)$ over Nudge$^1(K)$}
	\label{fig:ATIRKcont}
\end{subfigure}
\begin{subfigure}{.48\textwidth}
  \centering
  \includegraphics[width=1\linewidth]{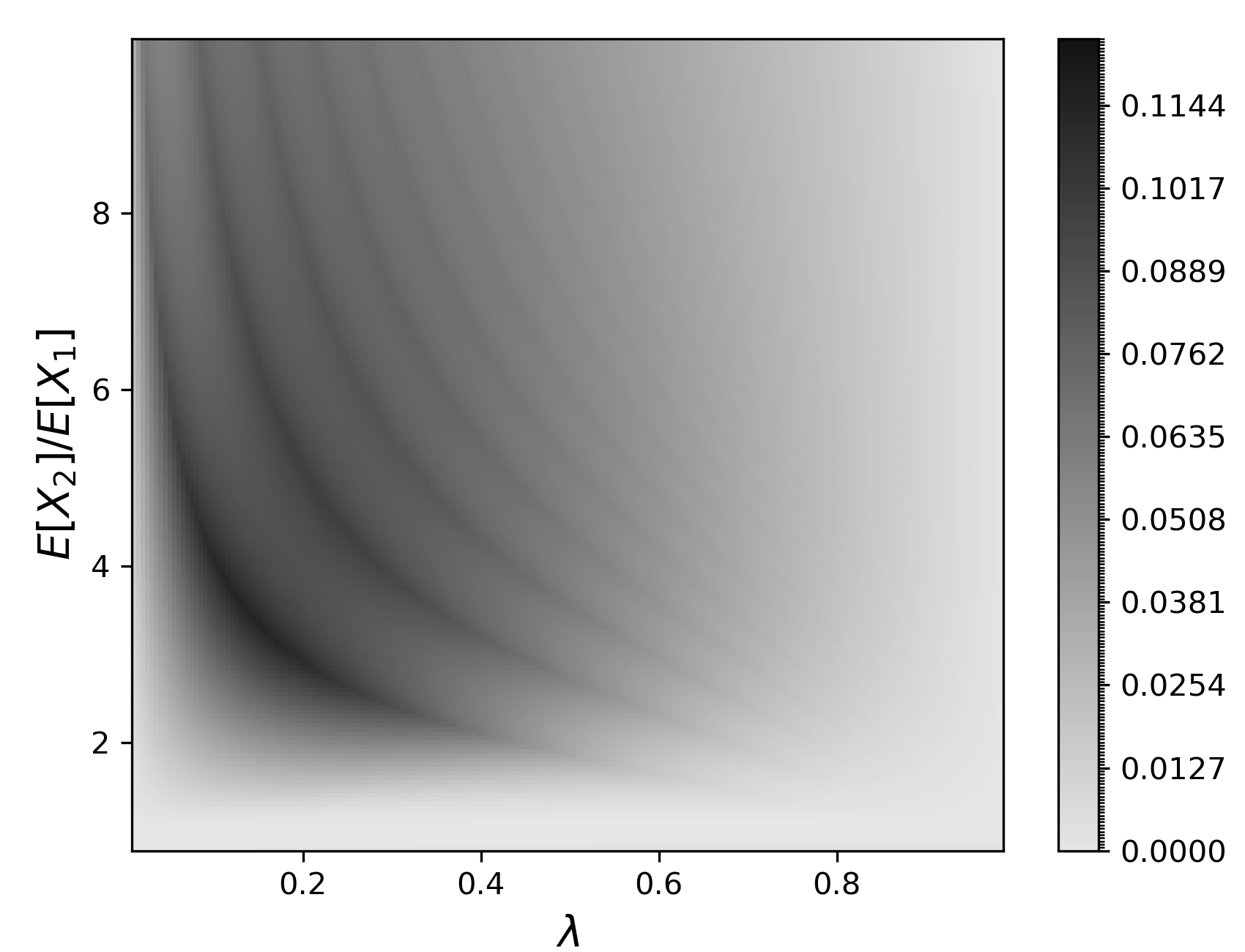}
	\caption{Plot of ATIR for $\gamma$-Boost over Nudge*$(M)$}
	\label{fig:ATIRMcont}
\end{subfigure}
\caption{The asymptotic tail improvement ratio of Nudge*$(M)$ over Nudge$^1(K)$ (left) and $\gamma$-Boost over Nudge*$(M)$ (right) with optimized parameters for $p=0.9$ with
exponential type-1 and type-2 jobs. Nudge*$(M)$ yields a much higher
ATIR than Nudge$^1(K)$, especially for high loads $\lambda$. $\gamma$-Boost improves Nudge*$(M)$ by exploiting arrival
time information, but the gains vanish
as the load tends to one. The flat white area in the first plot corresponds to an ATIR of zero which occurs when
$M_{opt} \leq 1$.}
\label{fig:ATIRcomp}
\end{figure*}

\section{Heavy Traffic Regime}\label{sec:heavy}

In this section we make several contributions. We present a simple closed form formula for the
ATIR of Nudge*$(M_{opt})$ as $\lambda$ tends to one. We show that this simple formula coincides with
the ATIR of $\gamma$-Boost, which shows that in the heavy traffic regime both algorithms perform equally
well. This implies that the additional arrival time information used by $\gamma$-Boost becomes redundant.
This can be understood intuitively as follows. When the load is low $\gamma$-Boost states that it is optimal
for a type-$1$ job that arrives at time $t$
to pass any type-$2$ job that arrived in $(t-b,t)$, for some fairly small $b$. The number of arrivals in this interval is
not well approximated by some deterministic number $M$. However when the load tends to one, the boost value $b$
of $\gamma$-Boost becomes very large (and tends to infinity) and therefore the relative error of approximating
the number of arrivals in $(t-b,t)$
by some fixed $M$ tends to zero.

We will end this section by deriving an expression for the optimal $M$ value in heavy traffic in terms of $E[X_1], E[X_2], E[X^2]$ and $\lambda$,
which we denote as $M_{heavy}$. We numerically demonstrate that this value is also a good approximation for
$M_{opt}$ when $\lambda$ is not close to one.

\begin{theorem}\label{th:heavy}
    For Nudge*$(M)$ with $E[X_2]\geq E[X_1]$, we have 
    \begin{align}
        \lim_{\lambda \rightarrow 1^-} ATIR^*(M_{opt}) &= 1 - p \left(\frac{E[X_1]}{E[X_2]}\right)^{(1-p)E[X_2]} - (1-p) \left(\frac{E[X_2]}{E[X_1]}\right)^{pE[X_1]}, \nonumber\\
        & = 1 - \left( \frac{1}{E[X_1]}\right)^{pE[X_1]}  \left( \frac{1}{E[X_2]}\right)^{(1-p)E[X_2]}
        \label{eq:ATIRheavy}
    \end{align}
   where
   \begin{align}\label{eq:Mheavy}
       M_{opt} \approx \left\lfloor \frac{\log(E[X_2]/E[X_1])}{\log(1+\theta_Z)}  \right\rfloor \approx \left\lfloor \frac{\log(E[X_2]/E[X_1])E[X^2]}{2(1-\lambda)}  \right\rfloor 
   \end{align}
   for $\lambda$ close to one.
   Moreover, this limit is increasing in $E[X_2]/E[X_1]$ on $(1,\infty)$, equals $0$ in $1$ and
    \begin{align}
        \lim_{E[X_2]/E[X_1] \rightarrow \infty} \lim_{\lambda \rightarrow 1^-} ATIR^*(M_{opt}) = p. 
    \end{align}
   
\end{theorem}
\begin{proof}{Proof:}
    The expression for $M_{opt}$ for $\lambda$ close to one is immediate from  \cite[Theorem 11]{vanhoudt_nudge} (and its proof) as
    both Nudge*$(M)$ and Nudge$^1(K)$ have the same optimal parameter value $M_{opt}$.
    To compute the limit of the $ATIR^*(M_{opt})$ we make use of the fact that $\tilde S(-\theta_Z) = 1 + \theta_Z + o(\theta_Z)$
    and similarly that $\tilde S_i(-\theta_Z) = 1 + E[X_i] \theta_Z + o(\theta_Z)$. Note that $\theta_Z$ tends to zero as $\lambda$ tends to one.

    We first look at the limit of $\tilde S(-\theta_Z)^{M_{opt}}$ as $\lambda$ tends to one:
    \begin{align}
        \lim_{\lambda \rightarrow 1^-} \tilde S(-\theta_Z)^{M_{opt}} &= \lim_{\theta_Z \rightarrow 0^+} (1+\theta_Z)^{\frac{\log(E[X_2]/E[X_1])}{\log(1+\theta_Z)}}
    \nonumber \\ &= e^{\log(E[X_2]/E[X_1])} = E[X_2]/E[X_1], 
    \end{align} 
    as $\lim_{x \rightarrow 0} (1+ax)^{b/\log(1+x)} = e^{ab}$, which we apply with $a=1$ and $b=\log(E[X_2]/E[X_1])$.
    Next consider the limit of $(w+w_1)^{M_{opt}}$ as $\lambda$ tends to one. First note that
    $w_1 \approx p(1+E[X_1]\theta_Z)/(1+\theta_Z)$ and $w \approx (1-p)/(1+\theta_Z)$, meaning
    $w_1+w \approx (1+pE[X_1]\theta_Z)/(1+\theta_Z)$. This yields
    \begin{align}
        \lim_{\lambda \rightarrow 1^-} (w_1+w)^{M_{opt}} &= \lim_{\theta_Z \rightarrow 0^+} \left(\frac{(1+pE[X_1]\theta_Z)}{(1+\theta_Z)}\right)^{\frac{\log(E[X_2]/E[X_1])}{\log(1+\theta_Z)}}
    \nonumber \\ &= e^{pE[X_1]\log(E[X_2]/E[X_1])} E[X_1]/E[X_2] 
    \nonumber \\ &= (E[X_1]/E[X_2])^{1-pE[X_1]}, 
    \end{align} 
    where we applied $\lim_{x \rightarrow 0} (1+ax)^{b/\log(1+x)} = e^{ab}$ twice.
    The result now follows by noting that $pE[X_1]+(1-p)E[X_2]=1$ and
    \[ ATIR^*(M_{opt}) = 1 - w_1 (w_1+w)^{M_{opt}} - (1-w_1) (w_1+w)^{M_{opt}} \tilde S(-\theta_Z)^{M_{opt}},\]
    where $w_1$ tends to $p$ as $\lambda$ goes to one.

    To compute the limit as $E[X_2]/E[X_1]$ tends to infinity, denote $E[X_2]/E[X_1]$  as $x$. 
    As $1=pE[X_1]+(1-p)E[X_2]$, we can express both $E[X_1]$ and $E[X_2]$ in terms of $x$, leading to
    \[  \lim_{\lambda \rightarrow 1^-} ATIR^*(M_{opt}) = 
     1-p x^{\frac{-(1-p)}{(1-p)x+p}}-(1-p)
    x^{\frac{p}{(1-p)x+p}}. \]
    The derivative of this function in $x$ factorizes as 
    \[ \frac{d}{dx} \lim_{\lambda \rightarrow 1^-} ATIR^*(M_{opt}) =  
    \frac{p(1-p)}{(1-p)x+p} \log(x) x^{\frac{-(1-p)x}{(1-p)x+p}},\]
    which is positive for $x > 1$. Further,
    \[   \lim_{x \rightarrow \infty} \lim_{\lambda \rightarrow 1^-} ATIR^*(M_{opt}) = 
    \lim_{x \rightarrow \infty} 1-p x^{\frac{-(1-p)}{(1-p)x+p}}-(1-p)
    x^{\frac{p}{(1-p)x+p}} = p, \]
    as the second term tends to zero and the last to $1-p$.
\Halmos \end{proof}

\begin{cor}\label{cor:ATIRsame}
Let ATIR$_\gamma$ denote the ATIR of $\gamma$-Boost, then
\[\lim_{\lambda \rightarrow 1^-} ATIR^*(M_{opt}) = \lim_{\lambda \rightarrow 1^-} ATIR_\gamma, \]
meaning the ATIR of $\gamma$-Boost and Nudge*$(M)$ coincide in the heavy traffic limit.
\end{cor}
\begin{proof}{Proof:}
By simplifying the expression for the tail constant $C_{Boost}$ in \citep{boostZiv} in the setting of two job types
one finds:
\begin{align} % ATIR(\mbox{$\gamma$-Boost})
    ATIR_\gamma = 1 - \frac{\tS{}-1}{\tS{}} \left(\frac{\tS{1}}{\tS{1}-1}\right)^{
    \frac{\lambda p (\tS{1}-1)}{\theta_Z}}
    \left(\frac{\tS{2}}{\tS{2}-1}\right)^{\frac{\lambda (1-p) (\tS{2}-1)}{\theta_Z}}.
\end{align}
When $\lambda$ tends to one, $\theta_Z$ goes to zero, while $\tS{i}$  and $\tS{}$ converge to one.
For $\lambda$ close to one, $\tS{i}$ can be approximated by $1+\theta_Z E[X_i]$ 
as $\tS{i}=1+\theta_Z E[X_i]+o(\theta_Z)$ due to the Taylor series expansion of $\tS{i}$.
In the same manner $\tS{}$ approaches $1+\theta_Z$ for $\lambda$ close to one.  
Hence,
\begin{align*}
    \lim_{\lambda \rightarrow 1^-} ATIR_\gamma &= 1 - \lim_{\theta_Z \rightarrow 0^+}  \frac{\theta_Z}{1+\theta_Z} 
    \left(\frac{1+\theta_Z E[X_1]}{\theta_Z E[X_1]}\right)^{p E[X_1]}
    \left(\frac{1+\theta_Z E[X_2]}{\theta_Z E[X_2]}\right)^{(1-p) E[X_2]} \\
    &= 1 - \lim_{\theta_Z \rightarrow 0^+}  
    \left(\frac{1+\theta_Z E[X_1]}{E[X_1]}\right)^{p E[X_1]}
    \left(\frac{1+\theta_Z E[X_2]}{E[X_2]}\right)^{(1-p) E[X_2]},
\end{align*}
which equals \eqref{eq:ATIRheavy}.
\Halmos \end{proof}

Theorem \ref{th:heavy} has two more important consequences:
\begin{enumerate}
    \item It shows that the heavy traffic limit of the $ATIR^*(M_{opt})$ is insensitive to the shape of the job size distributions $X_1, X_2$ and $X$.
 \item It further shows that for any $\epsilon > 0$ there exists an $X$, $X_1$ and $X_2$ such that
 for $\lambda$ sufficiently close to one we have $c_M/c_Z \leq \epsilon$, that is, 
 \[ \lim_{t \rightarrow \infty} e^{\theta_Z t} P[W_M > t] \leq \epsilon c_Z.\]
 This means that there is no upper bound on the improvement that can be achieved in the prefactor of FCFS. 
 To see this, one can set $p=1-\epsilon/2$ and use the optimal $M$ such that
 \[     \lim_{E[X_2]/E[X_1] \rightarrow \infty} \lim_{\lambda \rightarrow 1^-} 
 \left( 1- \lim_{t \rightarrow \infty} e^{\theta_Z t} P[W_M > t]/c_Z \right) = 1-\epsilon/2.\]
\end{enumerate}
We further note that the optimal $M$ for the heavy traffic regime given by \eqref{eq:Mheavy},
which we denote as $M_{heavy}$, 
is expressed in terms of $\lambda, E[X_1], E[X_2]$ and $E[X^2]$ only and may therefore be easier to estimate
in practice than \eqref{eq:Kopt}. We therefore also compare the ATIR of Nudge*$(M)$ with $M=M_{opt}$
and $M=M_{heavy}$ in Figure \ref{fig:HT}. The figure illustrates that using $M_{heavy}$ is also
effective when the load is high, but not necessarily very close to one. 

\begin{figure*}[t!]
\begin{subfigure}[t]{.48\textwidth}
  \centering
  \includegraphics[width=1\linewidth]{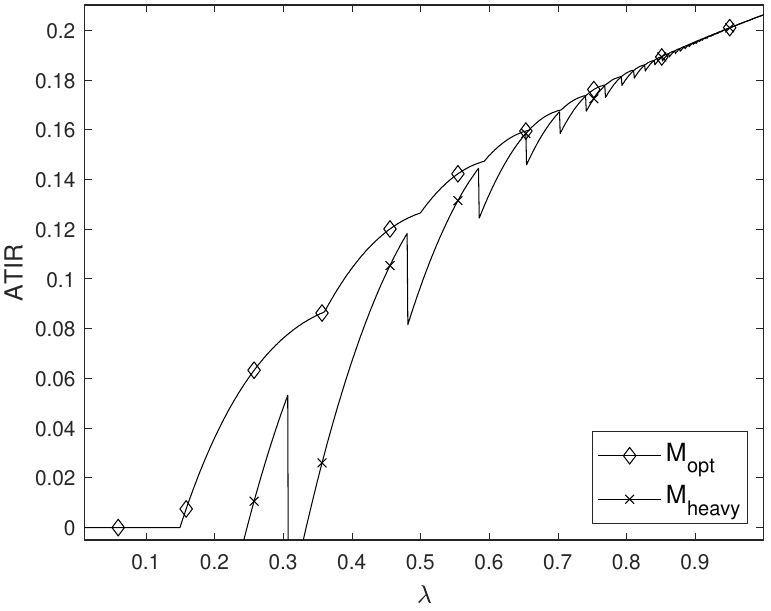}
	\caption{Exponential type-1 jobs and type-2 jobs}
	\label{fig:HTexp}
\end{subfigure}
\begin{subfigure}[t]{.48\textwidth}
  \centering
  \includegraphics[width=1\linewidth]{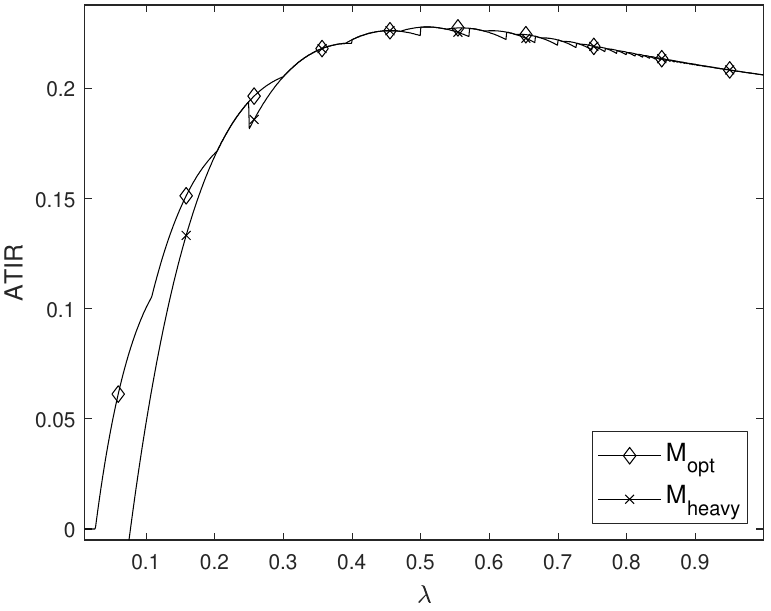}
	\caption{Exponential type-1 jobs, Hyperexponential ($SCV = 2$, $f = 1/2$) type-2 jobs}
	\label{fig:HTH2}
\end{subfigure}
\caption{The asymptotic tail improvement ratio of Nudge*$(M_{opt})$ over FCFS and of Nudge*$(M_{heavy})$ over FCFS with with $p=2/3$ and $E[X_1]/E[X_2]=4$.
Setting $M=M_{heavy}$ yields most of the
gains even when $\lambda$ is not very close to one.}
\label{fig:HT}
\end{figure*}

\section{Optimality of Nudge*$(M)$ within a family of Nudge policies}\label{sec:opt}

The optimality of Nudge*$(M)$ with $M=M_{opt}$ in the family $\mathcal{F}$ is established in a number of steps. 
We first present an expression for $c_{W^{(1)}}$ and $c_{W^{(2)}}$ of any scheduling algorithm in $\mathcal{F}_M$
in Lemma \ref{th:CWgen}.
Next we derive a simple condition in Lemma \ref{th:ns} such that if an algorithm $A \in \mathcal{F}$ is modified by increasing $n(s)$
by one and the new function $n$ corresponds to a feasible algorithm, then the ATIR increases if and only if this condition holds.
Finally, in Theorem \ref{th:improvedATIR} we show that starting from any algorithm $A \in \mathcal{F}$, we can find a sequence of algorithms 
in $\mathcal{F}$ that improve the ATIR and that transform $A$ into the Nudge*$(M)$ algorithm, which proves the optimality. 

Increasing $n(s)$ by one for an arbitrary arrival sequence $s$ does not always result in a feasible policy. For instance 
consider the FCFS algorithm (meaning all $n(s)$ are zero) and increase $n(12)$ to one, while all other $n(s)$ remain zero. 
Now assume we have three consecutive arrivals such that the first is type-$2$ and the next two are type-$1$, then the type-$1$
job that arrived last should be served after the other type-$1$ job, but before the type-$2$ job (as $n(12)=1$).
The type-$1$ job that arrived first however should be served after the type-$2$ job (as $n(21)=n(22)=0$). 
Clearly there is no order in which these three jobs can be served without violating one of these requirements
and therefore the function $n$ with $n(s)=0$ for all $s\not= 12$ and $n(12)=1$ is not feasible. This can also be seen as
it violates condition (C2) with $M=2$, $s_0=1$, $s_1=2$ and $s_2=1$ or $2$. In this case, Lemma \ref{th:ns} cannot be applied as the resulting function $n$ does not correspond to a feasible policy.

\begin{lemma}\label{th:CWgen}
    Consider a scheduling algorithm belonging to $\mathcal{F}_M$ characterized by the function $n$, then the prefactor
    of the type-1 and type-2 waiting time can be expressed as:
\begin{align}\label{eq:FCW1}
    c_{W^{(1)}} = \frac{c_Z}{\tilde S(-\theta_Z)^M} 
    \sum_{s \in \{1,2\}^M} (1-p)^{t(s)} p^{M-t(s)} \tilde S_1(-\theta_Z)^{M-t(s)} \tilde S_2(-\theta_Z)^{t(s)-n(s)}.
\end{align}
and
\begin{align}
   c_{W^{(2)}} = \frac{c_Z}{\tilde S(-\theta_Z)^{M-1}} 
    \sum_{\substack{s \in \{1,2\}^{2M} \\ s_{M+1}=2}} & \frac{(1-p)^{t(s)} p^{2M-t(s)}}{(1-p)} \tilde S_1(-\theta_Z)^{M-1-t(s_{M+2}\dots s_{2M})}\tilde S_2(-\theta_Z)^{t(s_{M+2}\dots s_{2M})}  \nonumber\\
    &\hspace*{1cm}  \cdot \prod_{k=1}^{M} \tilde S_1(-\theta_Z)^{1(s_k = 1 \wedge n(s_{k+1}\ldots s_{k+M})>t(s_{k+1}\ldots s_M))}. 
    \label{eq:FCW2}
\end{align}
\end{lemma}
\begin{proof}{Proof:}

We provide a sketch of the proof similar to the previous Lemmas and refer to Section \ref{app:CWgen} for a detailed proof.
For the type-1 jobs one first removes the work of the last $M$ jobs, which yields $c_Z/\tilde S(-\theta_Z)^M$.
Then we condition on the types of these last $M$ arrivals, given by the string $s$, and add the work of the
$M-t(s)$ type-1 jobs as well as the work of the $t(s)-n(s)$ jobs as these are not passed.

For type-$2$ jobs we need to look at the type of the next $M$ arrivals, but also at the types of the previous $M-1$
arrivals. We first remove the last $M-1$ arrivals and add the work of these $M-1$ jobs back conditioned on their types, that is,
we use $(p \tilde S_1(-\theta_Z) +(1-p) \tilde S_2(-\theta_Z))^{M-1}=\tilde S(-\theta_Z)^{M-1}$. 
For each of the $2^{M-1}$ terms in this $(M-1)$-th power we condition on the next $M$ arrivals and for each of these
$M$ arrivals we check whether the arrival is type-$1$ and passes the type-$2$ job. 

\Halmos \end{proof}

\begin{lemma}\label{th:ns}
    Let $s \in \{1,2\}^M$ selected arbitrarily, if we increase $n(s)$ by one such that the
    conditions (C1) and (C2) on the function $n$ remain valid, then this leads to an improved ATIR if and only
    if the position of the $n(s)+1$-st two in $s$ is in the first $M_{opt}$ positions.
\end{lemma}
\begin{proof}{Proof:}
    
      A sketch of the proof is given, for a detailed proof we refer to Section \ref{app:ns}. By definition 
      the ATIR improves when increasing $n(s)$ by one if and only if
    \[p \tilde S_1(-\theta_Z)  \Delta c_{W^{(1)}} + (1-p) \tilde S_2(-\theta_Z) \Delta c_{W^{(2)}}< 0,\]
    where $\Delta c_{W^{(i)}}$ is the change in the constants $c_{W^{(i)}}$ given in Theorem \ref{th:CWgen}, for $i=1,2$.
    The change in $c_{W^{(1)}}$ is easy to express as only one term in \eqref{eq:FCW1} changes. For 
    $c_{W^{(2)}}$ the change is more involved to express as multiple terms in \eqref{eq:FCW2} are affected
    by the increase in $n(s)$. After obtaining these changes, one plugs them into the above condition to find that
    the ATIR improves if and only if
    \begin{align}\label{eq:improve?}
    \frac{1}{\tilde S(-\theta_Z)^{k'}}\frac{\tilde S_1(-\theta_Z)(\tilde S_2(-\theta_Z)-1)}{\tilde S_2(-\theta_Z)(\tilde S_1(-\theta_Z)-1)} > 1,
    \end{align}
    which proves the stated result. \Halmos
    
\end{proof}

\begin{theorem}\label{th:improvedATIR}
    The Nudge*$(M)$ scheduling algorithm with $M=M_{opt} \geq 0$ has the highest ATIR in $\mathcal{F}$.
    Further, Nudge*$(M)$ maximizes the ATIR in $\mathcal{F}_M$  for $M \leq M_{opt}$. 
   \end{theorem}
\begin{proof}{Proof:}
    We first prove the optimality of Nudge*$(M)$ in $\mathcal{F}_M$ for $M \leq M_{opt}$. This follows from
    the fact that any scheduling algorithm in $\mathcal{F}_M$ 
    can be transformed into Nudge*$(M)$ by a sequence of operations that increase $n(s)$ by one for some string $s$ without violating the conditions (C1) and (C2) on the function $n$. More specifically, we repeatedly pick $s$ such that the following three conditions
    hold (until $n(s)=t(s)$ for all $s$):
    \begin{itemize}
        \item $n(s) < t(s)$,
        \item $n(s)$ is minimal among all strings $s$ with $n(s) < t(s)$,
        \item the position $k'$ that holds the $n(s)+1$-st two is minimal among all strings satisfying the previous
        two conditions.
    \end{itemize}
Condition (C1) is clearly preserved during this procedure. To see that condition (C2) remains true, pick $s$
such that the above three conditions hold and define $\bar n(s')=n(s')$ for all $s' \not=s$ and let $\bar n(s)= n(s)+1$.
To show that $\bar n$ satisfies (C2), given that $n$ satisfied (C2), we must check two cases: the case where $s$
appears on the left in (C2) and the case where $s$ appears on the right. First, assume $s=s_0 \ldots s_{M-1}$
and we need to check (C2) for $\bar n$. This means
\[ \bar n(s) = n(s) +1 \leq \bar n(s_1 \ldots s_{M-1} s_{M}) + 1(s_0 =2), \]
must hold. As $n(s)$ was minimal, the only problem that can occur is when $\bar n(s_1 \ldots s_{M-1} s_{M}) = n(s_1 \ldots s_{M-1} s_{M}) =n(s)$ and $s_0=1$. However, in such case $s_1 \ldots s_{M}$ would have been selected instead 
of $s=s_0 \ldots s_{M-1}$ as it has the same $n$ value, but its $n(s)+1$-st two appears in position $k'-1$. 
The second case we must verify is whether
\[ \bar n(s_0 \ldots s_{M-1}) \leq  \bar n(s) + 1(s_0 =2), \]
holds where $s=s_1 \ldots s_M$.
The left hand side equals $n(s_0 \ldots s_{M-1})$ and the right $n(s)+1+1(s_0=2)$, so this clearly holds
as (C2) was valid for $n$. As $k' \leq M_{opt}$, this change to $n(s)$ increases the ATIR due to Lemma \ref{th:ns}.

We can then prove that Nudge*$(M)$ with parameter $M_{opt}$ is optimal in $\mathcal{F_M}$ with $M > M_{opt}$ using a similar transformation. Any scheduling algorithm in $\mathcal{F_M}$ with $M > M_{opt}$ can be transformed into Nudge*$(M)$ with $M = M_{opt}$ by a sequence of operations where $n(s)$ is either incremented or decremented and where each operation improves the ATIR. We first repeat the previous part of the proof to increase $n(s)$ for some strings $s$ such that all twos in positions $k' \leq M_{opt}$ can be passed by an arriving type-1 job. This improves the ATIR and preserves conditions (C1) and (C2) as shown above. Secondly, $n(s)$ needs to be decreased for some strings $s$ to make sure no twos beyond position $M_{opt}$ can be passed
by an arriving type-$1$ job. This can be done by repeatedly picking $s = s_1 \ldots s_M$ such that the following three conditions hold (until no more strings satisfy them, i.e. $n(s) = t(s_1 \ldots s_{M_{opt}})$ for all $s$):
\begin{itemize}
    \item $n(s) > t(s_1 \ldots s_{M_{opt}})$, as this implies a type-2 job beyond position $M_{opt}$ can be passed,
    \item $n(s)$ is maximal among all strings $s$ with $n(s) > t(s_1 \ldots s_{M_{opt}})$,
    \item the position $k'$ that holds the $n(s)$-st two is maximal among all strings satisfying the previous two conditions.
\end{itemize}
As $n(s)$ stays the same for some $s$ while decreasing for other strings, condition (C1) clearly stays preserved. To see that condition (C2) also remains true, pick a string $s$ satisfying the previous three conditions and define $\bar{n}(s') = n(s')$ for all $s' \neq s$ and $\bar{n}(s) = n(s)-1$. We can again show that $\bar{n}$ satisfies (C2) given that $n$ did by checking the two cases where $s$ either appears in the left-hand side of (C2), or in the right-hand side. For the first case, assume $s = s_0 \ldots s_{M-1}$. This means
%$$\bar{n}(s) = n(s) - 1 \leq \bar{n}(s_1 \ldots s_M) + 1(s_0 = 2) = n(s_1 \ldots s_M) + 1(s_0 = 2)$$
$$\bar{n}(s) \leq \bar{n}(s_1 \ldots s_M) + 1(s_0 = 2)$$
must hold. Clearly this is the case given that $n(s) \leq n(s_1 \ldots s_M) + 1(s_0 = 2)$ holds, as $\bar{n}(s) = n(s) - 1$ while the right-hand side is the same in both equations. The second case assumes $s = s_1 \ldots s_M$ and we must verify whether
$$\bar{n}(s_0 \ldots s_{M-1}) \leq \bar{n}(s) + 1(s_0 = 2)$$
holds. When $s_0 = 2$, this inequality simplifies to $\bar{n}(s_0 \ldots s_{M-1}) \leq \bar{n}(s) + 1 = n(s)$, which holds as $n(s)$ was maximal. When $s_0 = 1$, the inequality can be rewritten as $\bar{n}(s_0 \ldots s_{M-1}) = n(1 s_1 \ldots s_{M-1}) \leq \bar{n}(s) = n(s)-1$, or $n(1 s_1 \ldots s_{M-1}) < n(s)$. As $n(s)$ was maximal, this can only cause a problem if $n(1 s_1 \ldots s_{M-1}) = n(s)$. However, in that case the string $1 s_1 \ldots s_{M-1}$ would be selected before $s = s_1 \ldots s_M$ as it has the same $n$ value while its $n(s)$-st two appears in position $k' + 1$. Note that $n(s) > t(s_1 \ldots s_{M_{opt}})$ and $n(s) = n(1 s_1 \ldots s_{M-1})$ imply that $n(1 s_1 \ldots s_{M-1}) > t(1 s_1 \ldots s_{M_{opt}-1})$ as $t(s_1 \ldots s_{M_{opt}}) \geq t(1 s_1 \ldots s_{M_{opt}-1})$.

As $n(s) > t(s_1 \ldots s_{M_{opt}})$, the position $k'$ of the $n(s)$-st two in $s$ is beyond $M_{opt}$, which implies this change from $n$ to $\bar{n}$ improves the ATIR due to Theorem \ref{th:ns}.
\Halmos \end{proof}

Theorem \ref{th:ns} also allows us to establish the following structural results for the Nudge$^{*1}(K,M)$, 
Nudge$^{*2}(M,L)$ and Nudge$^{12}(K,L)$
scheduling algorithms:
\begin{itemize}
    \item For Nudge$^{*1}(K,M)$ or Nudge$^{*2}(M,L)$ with $M \leq M_{opt}$ increasing $K$ or $L$ always improves the ATIR.
    Given a $K$ fixed, the best value for $M$ is $\max(K,M_{opt})$. The same holds with respect
    to $L$.
    \item For the Nudge$^{12}(K,L)$ scheduling algorithm denote the optimal $K$ and $L$ as $(K^*,L^*)$, then
    $K^*+L^*-1 \in [M_{opt}, 2 M_{opt} - 1]$. The lower bound is due to the fact that otherwise increasing $K$ or $L$
    improves the ATIR and the upper bound follows from noting that $K^*,L^* \leq M_{opt}$. Explicit results
    for $K^*$ and $L^*$ appear hard to obtain.
\end{itemize}

Figure \ref{fig:ATIR4} illustrates Theorem \ref{th:improvedATIR} when both job sizes are exponential. It compares the
ATIR of five different scheduling algorithms with optimized parameters (for Nudge$^{12}(K,L)$ the optimization was done by
brute force). The figure shows significant gains in the ATIR for Nudge*$(M)$ over Nudge$^1(K)$ and Nudge$^2(L)$ when the load is 
sufficiently high. Note that four of the five scheduling algorithms coincide when $M_{opt} \leq 1$. The difference between Nudge$^{12}(K,L)$ and Nudge*$(M)$ is small in case $K$ and $L$ are optimized. Finally the $\gamma$-Boost
algorithm outperforms Nudge*$(M)$ as proven in \citep{boostZiv} by exploiting arrival time information, but the gain
tends to zero as the load tends to one (see Corollary \ref{cor:ATIRsame}).

\begin{figure*}[t!]
\begin{subfigure}{.48\textwidth}
  \centering
  \includegraphics[width=1\linewidth]{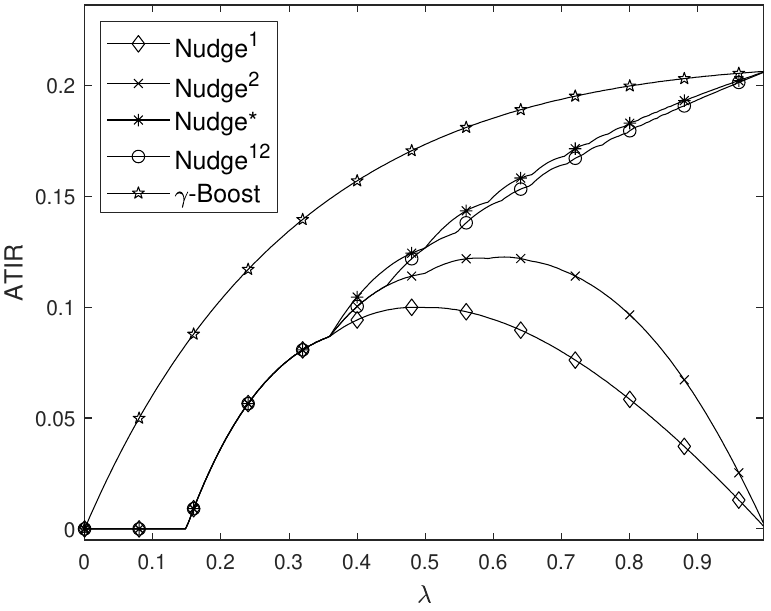}
	\caption{$E[X_2]/E[X_1]=4$, $p=2/3$}
	\label{fig:ATIRvs4}
\end{subfigure}
\begin{subfigure}{.48\textwidth}
  \centering
  \includegraphics[width=1\linewidth]{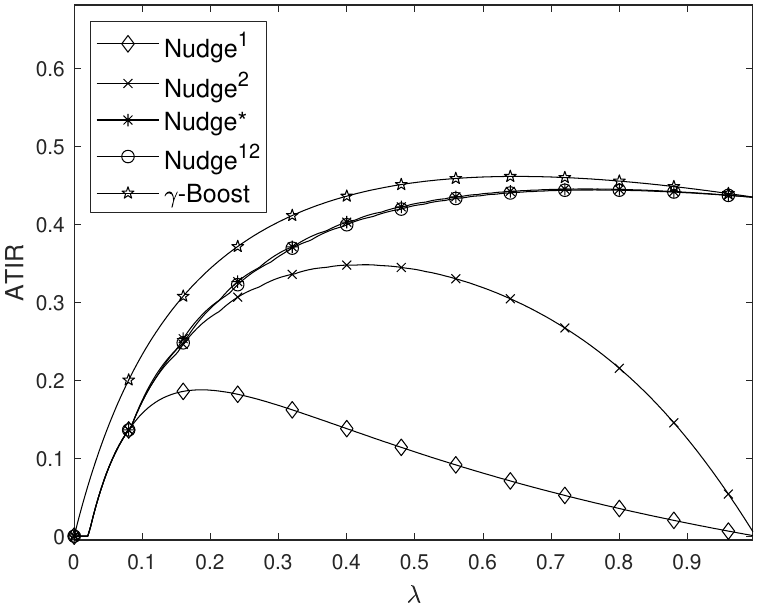}
	\caption{$E[X_2]/E[X_1]=10$, $p=9/10$}
	\label{fig:ATIRvs10}
\end{subfigure}
\caption{The asymptotic tail improvement ratio of Nudge$^1(K)$, Nudge$^2(L)$, Nudge$^{12}(K,L)$, Nudge*$(M)$, and $\gamma$-Boost (with optimized $K,L$, $M$, and boost), over FCFS with exponential job sizes.
Nudge*$(M)$ outperforms the other Nudge algorithms, but the difference with Nudge$^{12}(K,L)$ is small. $\gamma$-Boost outperforms the different Nudge algorithms by exploiting the arrival time information.}
\label{fig:ATIR4}
\end{figure*}

\section{Selecting Job Types} \label{sec:split}

\begin{figure*}[t!]
\begin{subfigure}[t]{.48\textwidth}
  \centering
  \includegraphics[width=\linewidth]{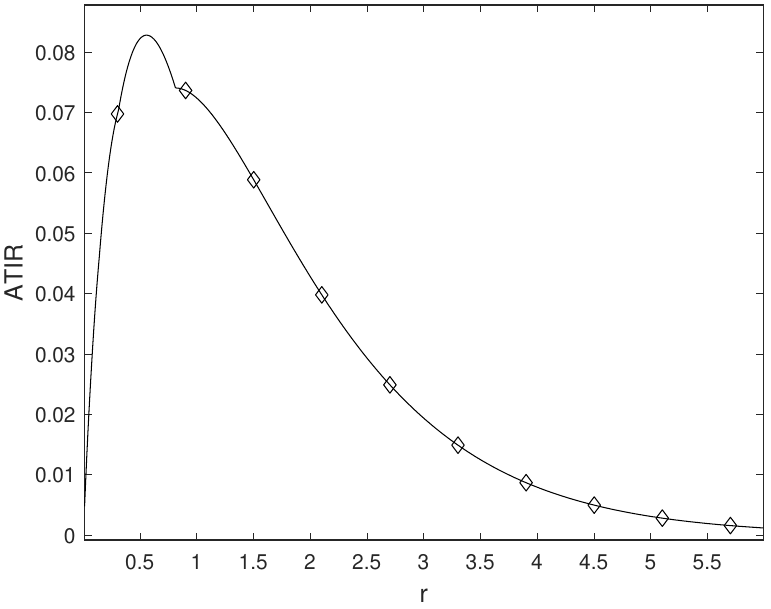}
  \caption{Exponential job sizes}\label{fig:7a}
\end{subfigure}
\begin{subfigure}[t]{.48\textwidth}
  \centering
  \includegraphics[width=\linewidth]{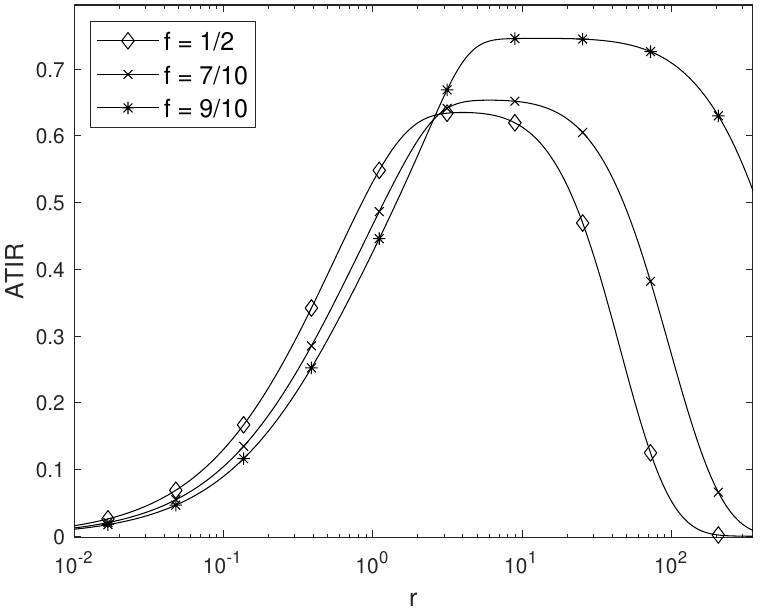}
  \caption{Hyperexponential job sizes}\label{fig:7b}
\end{subfigure}

\caption{The asymptotic tail improvement ratio (ATIR) of Nudge$^*(M_{opt})$ over FCFS using exponential or hyperexponential ($SCV = 10$ and various values of $f$) job sizes, $\lambda = 0.5$, and jobs split into two types based on a threshold $r$.}
\label{fig:ATIR_r}
\end{figure*}

In the previous sections we assumed that the jobs were already split into two types.
In this section we consider the problem of splitting the jobs into two types such that the tail of the response time
is minimized when the jobs are scheduled according to the Nudge*$(M_{opt})$ algorithm. As allowing shorter jobs to pass
longer jobs is essentially the only tool we have to improve the tail, we restrict ourselves to splitting the jobs
based on their size using a single threshold $r > 0$. Arriving jobs of size less than or equal to $r$ are type-1, while jobs larger than $r$ are type-2, that is, $X_1 = X | X \leq r$, $X_2 = X | X > r$, and $p = P[X \leq r]$.

Figure \ref{fig:ATIR_r} shows the ATIR of Nudge$^*(M_{opt})$ with the threshold $r$ varying along the x-axis
when $\lambda = 0.5$, where we note that $M_{opt}$ depends on $r$.  
Figure \ref{fig:7a} considers exponential job sizes, while Figure \ref{fig:7b} considers various hyperexponential distributions. 
As expected when $r$ approaches zero or becomes sufficiently large, the ATIR tends to zero (as nearly all the jobs 
have the same type). We further note that a near optimal ATIR is reached for a rather wide range of $r$ values, which
is especially true for hyperexponential job sizes with larger $f$. In such case the majority of the jobs are short and 
the tail improves by  passing the occasional large job.
Remark that these curves are not necessarily smooth in $r$ and may even have several local maxima as
demonstrated in Figure \ref{fig:8a}, which considers a setting with a load $\lambda = 0.95$. This implies that optimizing $r$
appears to be a challenging problem.

Even in the simple case with exponential job sizes, there does not appear to be a closed form formula for the optimal threshold $r$ for a given $M$, let alone the optimal combination $(r, M)$. 
More specifically, in this setting we have $\theta_Z = 1-\lambda$ and
\begin{align*}
    \tS{1} &= \int_0^r e^{(1-\lambda)t} \frac{e^{-t}}{1 - e^{-r}} \, dt = \frac{1-e^{-\lambda r}}{\lambda(1 - e^{-r})}, \\
    \tS{2} &= \int_r^\infty e^{(1-\lambda)t} \frac{e^{-t}}{e^{-r}} \, dt = \frac{e^{-(\lambda-1)r}}{\lambda}, \\
    \tS{} &= \int_0^\infty e^{(1-\lambda)t} e^{-t} \, dt = \frac{1}{\lambda},
\end{align*}
This leads to an ATIR given by
\begin{align*}
    %w &= \lambda e^{-r} \\
    %w_1 &= 1-e^{-\lambda r} \\
    ATIR^*(M, r) &= 1 - (1 + \lambda e^{-r} - e^{-\lambda r})^M (1-e^{-\lambda r} + e^{-\lambda r} \lambda^{-M}).
\end{align*}
The zero(s) of this function's derivative when choosing a fixed $M$ do not appear to have a closed form formula, and clearly this problem is only further complicated when replacing $M$ by $M_{opt}$ as it makes the derivative discontinuous in places where $M_{opt}$ changes.

To get some insight on which parameters affect the optimal choice of $r$, Figure \ref{fig:8b} includes a plot with the load on the x-axis and the SCV on the y-axis, where the shade of gray shows the best value of $r$. Two things stand out:
\begin{enumerate}
    \item For low loads, the color gradient is somewhat chaotic, as a change in $M_{opt}$ has a strong effect on the policy. This is not the case for higher loads as for example increasing $M_{opt}$ from 100 to 101 does not significantly change the policy.
    \item For high loads we observe a sharp (pixelated) line where the optimal $r$ value makes a jump from a low to a much
    higher value when the SCV becomes large enough. This somewhat unexpected jump can be understood by the occurrence of the multiple
    local maxima shown in Figure \ref{fig:8a}. It shows that for small SCV, the local maximum for $r \in [1,2]$ is also a global
    maximum, but as the SCV increases the local maximum associated to a larger $r$ value takes over. In fact, this means there
    exists a specific SCV value (somewhere between $5$ and $10$ in our example) such that there are two global maxima.
\end{enumerate}
Despite the complexity of this contour plot, a reasonable rule of thumb is that the optimal $r$ is affected less by the
arrival rate $\lambda$ and more by the variability of the job sizes, where higher variability tends to require a higher
optimal $r$. This makes sense as most benefit in case of highly variable job sizes, where most jobs are small, may come from passing 
the occasional large job.

\begin{figure*}[t!]
\begin{subfigure}[t]{.48\textwidth}
  \centering
  \includegraphics[width=\linewidth]{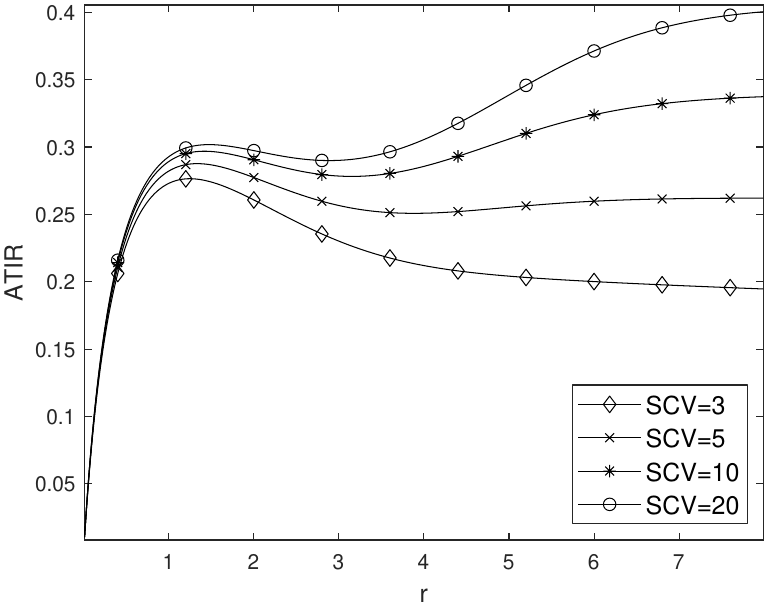}
  \caption{Hyperexponential job sizes, $\lambda = 0.95$}\label{fig:8a}
\end{subfigure}
\begin{subfigure}[t]{.5\textwidth}
  \centering
  \includegraphics[width=1\linewidth]{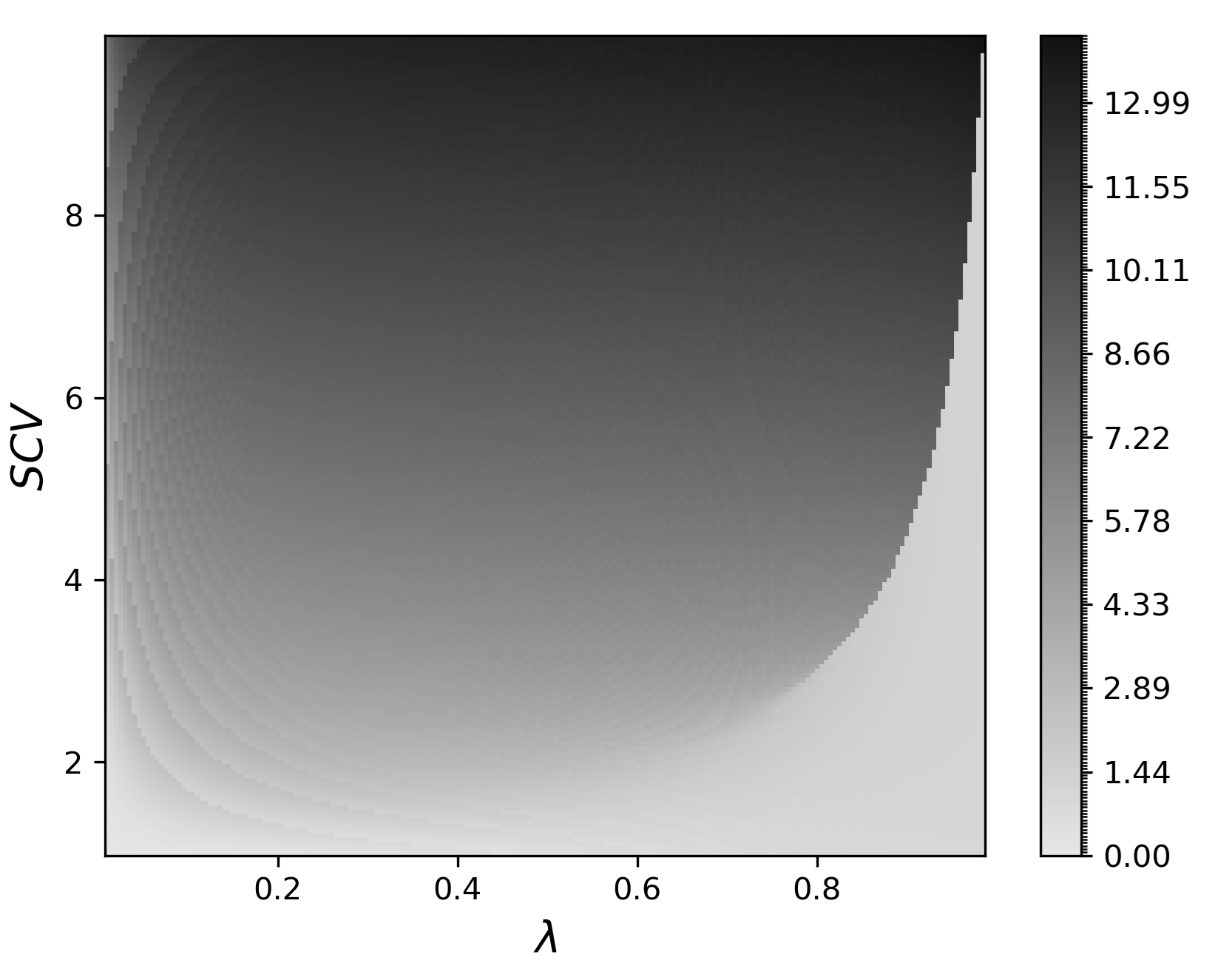}
  \caption{Plot of the optimal threshold $r$ of Nudge*$(M_{opt})$ for various loads and SCVs}\label{fig:8b}
  %\label{fig:ATIRMcont}
\end{subfigure}

\caption{Using hyperexponential ($f=9/10$) job sizes, the first plot shows that the ATIR of Nudge*$(M_{opt})$ in function of $r$ does not always have a single local maximum. The second plot shows the optimal threshold $r$ for various loads and SCVs.}
\label{fig:ATIR_r_contour}
\end{figure*}

\section{Nudge$^*(M)$ in a Multi-server Systems}\label{sec:multi}

\begin{figure*}[t!]
\begin{subfigure}[t]{.48\textwidth}
  \centering
  \includegraphics[width=\linewidth]{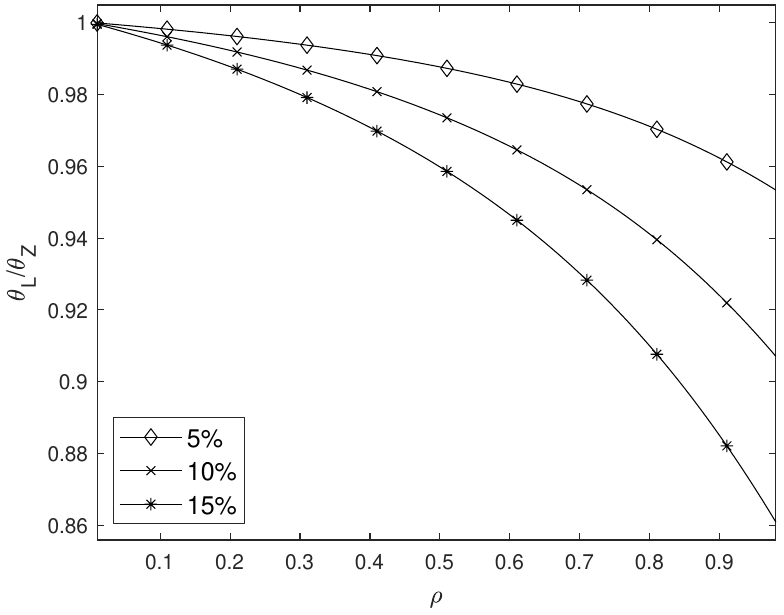}
  \caption{The decay rate of the large jobs becomes worse with dedicated servers for small jobs}\label{fig:9a}
\end{subfigure}
\begin{subfigure}[t]{.48\textwidth}
  \centering
  \includegraphics[width=1\linewidth]{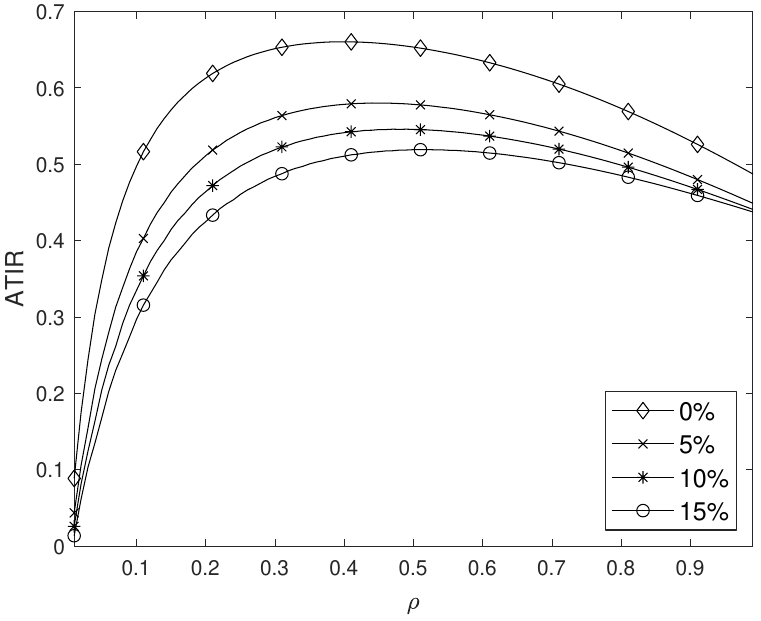}
  \caption{Without small jobs, Nudge*$(M)$ still achieves an improvement in the tail, but it does become smaller}\label{fig:9b}
  %\label{fig:ATIRMcont}
\end{subfigure}

\caption{Using hyperexponential ($SCV = 10$, $f = 7/10$) job sizes, the first plot shows that the decay rate of the workload can drop by a few percent compared to the baseline, when 5\%, 10\%, or 15\% of the queues are dedicated to the smallest jobs.
The second plot shows that the ATIR achieved by Nudge*$(M_{opt})$ in a server processing large jobs is slightly reduced compared to the baseline, but not significantly. In both plots, the x-axis shows the load $\rho$ of each queue.
} \label{fig:dedi_serv}
\end{figure*}

In the previous sections we focused on a single server system. This is equivalent to a multi-server system where each server
has its own waiting line and the incoming jobs are assigned at random among the servers upon arrival (as a probabilistic
split of a Poisson process is Poisson). 
In a multi-server setting it is
common practice to dedicate some servers to process small jobs to avoid that they get stuck in a waiting line behind a
long job (thereby reducing the mean response time), while the remaining jobs are randomly assigned among the other servers. As before all servers have
their own waiting line and jobs immediately join the waiting line of the selected server.
The objective of this section is to show that Nudge$^*(M)$ is still beneficial in such a setting despite the fact
that the shortest jobs have their own dedicated server(s).

%We also consider this job partitioning problem when dedicating a certain percentage of the queues to small jobs. The motivation behind this is that in practice, one might have one or more dedicated servers for small jobs, to try to avoid making them wait behind large jobs that aren't as latency sensitive. 
Assume that the dedicated servers serve all jobs of size less than some given threshold $m$. 
This threshold is chosen such that all servers have the same load, that is, the each dedicated server serving small jobs and any one of the remaining servers have the same load. For example with 10 servers where one is dedicated to small jobs, $m$ is chosen so that arriving jobs of size less than $m$ contribute 10\% of the total work, i.e., for a continuous distribution $X$,
\begin{align}
    \int_0^m t f_X(t) \, dt = \frac{1}{10}E[X].
\end{align}
Thus the dedicated server has the same load as in the setting of the random job assignment, but now only serves 
jobs of size $a \leq m$. The nine remaining servers also have the same load as before as the jobs of size $a > m$ contribute 90\% of the total work, which is evenly split among the nine remaining servers.

We first note that introducing dedicated servers affects the decay rate of the response time distribution of the jobs.
Recall that both FCFS and Nudge$^*(M)$ achieve the optimal decay rate $\theta_Z$ (which is the decay rate of the workload).
More specifically, compared to the baseline where all jobs are sent to a random server, the decay rate in the servers dedicated to small jobs improves (under FCFS or a Nudge policy), while the queues processing larger jobs experience a worse decay rate. This is shown in Figure \ref{fig:9a}, which plots the ratio between the decay rate in a dedicated queue $\theta_L$, and the decay rate of the baseline $\theta_Z$.  In other words, introducing dedicated servers for small jobs worsens the tail of the overall response time distribution (especially in case of larger loads).  With this in mind, we focus on the queues serving the larger jobs as the tail of the response time of a random job is determined by their scheduling order. We look at the asymptotic gains Nudge$^*(M)$ can achieve in these servers. Due to the Poisson arrivals, we can study these servers in isolation, we only need to adapt
the job size distribution to exclude the jobs of size $a \leq m$ and to modify $\lambda$ such that the load remains the
same. As in the previous section we split jobs into two types using a threshold $r > m$, that is,
\begin{align*}
X_1 = X | m < X \leq r, \hspace{1cm} X_2 = X | X > r. %, \\
%p &= 1 - P[X_{orig} > r]/P[X_{orig} \geq a], & \hspace{1cm} \lambda &= \lambda_{orig} P[X_{orig} \geq a].
\end{align*}
We further remark that Nudge$^*(M)$ not only improves the tail, but also the mean response time as it only switches 
jobs with a size between $m$ and $r$ with jobs that exceed $r$ in size. 

Figure \ref{fig:9b} shows that using Nudge*$(M)$ with optimized $r$ and $M$ still yields significant gains in the tail of the response time distribution, even when a small percentage of the servers is dedicated to small jobs. Therefore even in a system with some size-based dispatching, Nudge*$(M)$ still achieves a substantial improvement in  the tail of the response time.

\section{Response time distribution of type-2 jobs}\label{sec:t2}

In this section we develop a numerical method to compute the response time distribution of a type-2 job under Nudge*$(M)$, for phase-type job sizes. This is considerably more involved than in \citep{vanhoudt_nudge} for Nudge$^1(K)$ as in Nudge$^1(K)$
at most one job can pass a tagged type-2 job and this type-1 job must arrive within time $s$, where $s$ is the workload
observed by the tagged type-2 job upon arrival.

Even if we first computed the probabilities that a tagged type-2 job is swapped $k$ times with a 
type-1 job given that it observed a workload of $s$ upon arrival, we cannot directly use these
to compute the waiting time distribution of a tagged type-2 job. This is due to the fact that additional type-$1$ jobs
can also arrive while another type-$1$ job that passed the tagged type-$2$ is being served.
For this reason we need to introduce a more involved method to compute the distribution of the extra waiting time $W_{extra}(s)$  that a type-2 job experiences that sees a workload of $s$ upon arrival. 

We will argue that the distribution $W_{extra}(s)$ is a phase-type distribution 
characterized by $(\gamma(s), Q)$. To define $\gamma(s)$ and $Q$, we first introduce the square matrices $W_i$ of size $(i+1)$, for $i=1,\ldots,M$, as follows
\begin{align*}
     (W_i)_{k,\ell} = \left\{ 
    \begin{array}{ll}
        -\lambda  & \quad \ell = k < i+1, \\
        \lambda  & \quad \ell = k+1, \\
        0 & \quad otherwise.
\end{array}
\right.,
\end{align*}
These matrices simply count arrivals.
We also need the matrices 
\begin{align*}
 U^{(\ell)}_i &= p (1-p)^{\ell-1} \begin{bmatrix}
        0_{\ell \times (i+1-\ell)} \\
        I_{i+1-\ell}
    \end{bmatrix}. 
\end{align*}
These matrices are of size $(i+1)\times (i+1-\ell)$ and if we think of the row and column index as a counter, then $U_i^{(\ell)}$ simply
states that the counter decreases by $\ell$ with probability $p(1-p)^{\ell-1}$.
The subgenerator $Q$ is now defined by the following block upper triangular matrix:
\begin{align*}
Q &= \begin{bmatrix}
        W_{M-1} \oplus S_1 & \quad U_{M-1}^{(1)} \otimes s_1^*\alpha_1 & \quad U_{M-1}^{(2)} \otimes s_1^*\alpha_1 & \ldots & \quad U_{M-1}^{(M-1)} \otimes s_1^*\alpha_1 \\
        & W_{M-2} \oplus S_1 & \quad U_{M-2}^{(1)} \otimes s_1^*\alpha_1 & \ldots & \quad U_{M-2}^{(M-2)} \otimes s_1^*\alpha_1  \\
        &  & \ddots & \ddots & \vdots  \\
        &  &  & \quad W_{1} \oplus S_1 & \quad U_{1}^{(1)} \otimes s_1^*\alpha_1 \\
        &  &  &  & S_1 \\
\end{bmatrix}
\end{align*}
This subgenerator $Q$ consists of $M$ block rows, where the phases belonging to the $k$-th block row correspond to work that is added by the $k$-th arrival if it passes the tagged type-2 job. 
Each block row of this matrix contains two pieces of information: a counter $i \in \{0,\ldots,M-k\}$ and a service phase $j\in \{1,\ldots,n_1\}$.
The matrix $S_1$ captures the evolution of work added by the type-1 job and thus the phase $j$. The matrix $W_{M-k}$ increments 
the counter $i$ each time there is an arrival and halts in $M-k$ (there is no need to count further as
only $M$ arrivals can pass). When the work of the $k$-th arrival is completely added  (which occurs at the rates given by $s_1^*$), 
we move on to block row $k+\ell$ if the counter $i$ is at least $\ell$ and the types of the first $\ell$ arrivals are
$\ell-1$ type-$2$ jobs followed by a type-$1$ job. 
This is exactly what the matrix $U_{M-k}^{(\ell)}$ does: it removes the first $\ell$ jobs from the counter, and holds the probability $p (1-p)^{\ell-1}$ of having $\ell-1$ type-$2$ jobs followed by a type-$1$ job. When the counter $i=0$ or when all of the remaining jobs are type-$2$, the underlying Markov chain of the phase-type distribution moves to the implicit absorbing state. Note that when we are adding the work of the $M$-th arrival, the counter $i$ is zero and cannot change, so we only need to keep track of the phase of the work of the type-$1$ job.

The initial phase vector $\gamma(s)$ depends on the arrivals that occurred in the first interval of length $s$, which is
the workload that the tagged type-$2$ job observed upon arrival. It is given by
\begin{align}
    \gamma(s) &= e_1^* e^{W_M s} [U_M^{(1)}\otimes \alpha_1, U_M^{(2)}\otimes \alpha_1, \ldots, U_M^{(M)}\otimes \alpha_1]  \nonumber
    %\gamma(s) &= (\nu(s) U_M^{(1)} \otimes \alpha_1, \hspace{2mm} \nu(s) U_M^{(2)} \otimes \alpha_1, \hspace{2mm} \ldots, \hspace{2mm} \nu(s) U_M^{(M)} \otimes \alpha_1). \nonumber
\end{align}
The vector $e_1^*e^{W_M s}$ contains the probabilities of having $\ell = 0$ to $M$ arrivals during time $s$ (in fact the last
entry is the probability of having $M$ or more arrivals). In order for the extra
work to start in block row $k$ of $Q$, there must be at least $k$ arrivals, the first $k-1$ of which must be type-$2$ and the $k$-th 
must be type-$1$. The counter $i$ is initialized as $\ell -k$ and the service phase is initialized by the vector $\alpha_1$.
Notice that the vector $\gamma(s)$ does not sum to one, meaning
the phase-type distribution $(\gamma(s),Q)$ has a mass in zero, corresponding to the probability
that the tagged type-2 job is not swapped.
Given the above discussion it should now be clear that the extra waiting time for a type-2 job that sees a 
workload $s$ on arrival is given by the following formula:
\begin{align*}
    P[W_{extra}(s) > t] = \gamma(s) e^{Qt} \e
\end{align*}

We first recall the following result:
\begin{prop}[from \citep{vanhoudt_nudge}]\label{th:Z}
Let $Z$ be the workload in an M/PH/1 queue (for any work conserving scheduler), then
\begin{align}\label{eq:Zt}
P[ Z > t] = \lambda \alpha e^{T t}(-S)^{-1}\e = \lambda \beta e^{T t}(-T)^{-1}\e,
\end{align}
with $\beta = (1-\lambda)\alpha$ and $T = S + \lambda \e \alpha$.
\end{prop}

\begin{theorem}
    Let $W^{(2)}$ be the waiting time distribution of a type-2 job, then
     \begin{align*}
        P[W^{(2)} > t] = (e^*_1 \otimes \lambda \beta, 0) e^{T_M t} v_2,
    \end{align*}
    where % $\e_W$ is a column vector with the same size as $W_M$, 
    \begin{align*}
    %T_M &= \begin{bmatrix} W_M \oplus T & \quad U_M \otimes \e \alpha_1 \hspace{2mm} 0 \\ 0 & Q \end{bmatrix}
    T_M &= \begin{bmatrix} W_M \oplus T & \quad (I_{M+1}\otimes \e)V_M \\ 0 & \quad Q \end{bmatrix}
    \hspace*{5mm} \mbox{and}  \hspace*{5mm} v_2 = \begin{bmatrix} \e_{M+1} \otimes ((-T)^{-1} \e) \\ \e \end{bmatrix}
    \end{align*}
    with $V_M = [U_M^{(1)}\otimes \alpha_1, U_M^{(2)}\otimes \alpha_1, \ldots, U_M^{(M)}\otimes \alpha_1] $.

Let $R^{(2)}$ be the response time distribution of a type-2 job, then
    \begin{align}
        P[R^{(2)} > t] = P[W^{(2)} > t] + (1-\lambda) \alpha_2 e^{S_2 t} \e - (e^*_1 \otimes \lambda \beta, 0) e^{B t} \begin{bmatrix} \vecn \\ \e  \end{bmatrix},
    \end{align}
    where
    \begin{align*}
        B &= \begin{bmatrix}
            T_M & \quad  T_M v_2 \alpha_2 \\
            0 & S_2
        \end{bmatrix}.
    \end{align*}
\end{theorem}
\begin{proof}{Proof:}
   The waiting time for a type-2 job is larger than $t$ if the workload is larger than $t$, or if the workload is $s < t$ and the extra time $W_{extra}(s) > t-s$. The workload is 0 with probability $1-\lambda$, but in that case, the job goes into service immediately and so the waiting time is never larger than $t$. Hence,
    \begin{align*}
        P[W^{(2)} > t] &= P[Z > t] + \int_{0}^{t} P[W_{extra}(s) > t-s] (\lambda \beta e^{T s} \e) \,ds \\
        %&= P[Z > t] + \int_{0}^{t} \lambda \beta e^{T s} \e (\nu(s) U_M^{(1)} \otimes \alpha_1, \hspace{2mm} \ldots, \hspace{2mm} \nu(s) U_M^{(M)} \otimes \alpha_1) e^{Q(t-s)} \e \,ds \\
        &= P[Z > t] + \int_{0}^{t} e_1^* e^{W_M s} V_M e^{Q(t-s)} \e (\lambda \beta e^{T s} \e)  \,ds \\
        &= P[Z > t] + \int_{0}^{t} (e_1^* \otimes \lambda \beta) (e^{W_M s} \otimes e^{T s})((V_M e^{Q(t-s)} \e) \otimes \e)  \,ds \\
        &= P[Z > t] + \int_{0}^{t} (e_1^* \otimes \lambda \beta) e^{(W_M \oplus T) s}((I_{M+1} \otimes \e) V_M) e^{Q(t-s)} \e  \,ds \\
        %&= P[Z > t] + \int_{0}^{t} (e_1^* \otimes \lambda \beta) (e^{W_M s} \otimes e^{W_M s}) (\e \otimes (U_M \otimes \alpha_1,0)) e^{Q(t-s)} \e \,ds \\
        %&= P[Z > t] + \int_{0}^{t} (e_1^* \otimes \lambda \beta) e^{(W_M \oplus T)s} (( U_M \otimes \alpha_1,0) \otimes \e) e^{Q(t-s)} \e \,ds 
\end{align*}
Using Lemma 1 from \citep{vanhoudt_nudge} we get 
        \begin{align*}
        P[W^{(2)} > t]&= P[Z > t] +( e_1^* \otimes \lambda \beta, 0) e^{T_M t} \begin{bmatrix} \vecn \\ \e \end{bmatrix} = ( e_1^* \otimes \lambda \beta, 0)  e^{T_M t} \begin{bmatrix} \e_{M+1} \otimes ((-T)^{-1} \e) \\ \e \end{bmatrix},
    \end{align*}
    where the last equality holds as 
    \begin{align*}
     (e^*_1 \otimes \lambda \beta) e^{(W_M \oplus T)t} &(\e_{M+1} \otimes ((-T)^{-1} \e)) 
    = (e^*_1 e^{W_M t} \e_{M+1}) (\lambda \beta e^{T t} (-T)^{-1} \e) \\ 
    &= \lambda \beta e^{T t} (-T)^{-1} \e 
    = P[Z > t],
    \end{align*}
    as $e^*_1 e^{W_M t}$ is a row vector of probabilities that sum to 1.

    A type-2 job's response time consists of its waiting time combined with its service time. The response time is larger than $t$ if the waiting time is. If the waiting time is $s < t$, the response time may still be larger than $t$ if the job's service time is longer than $t-s$. The waiting time is 0 if the queue is empty with probability $1-\lambda$.
    \begin{align*}
        P[R^{(2)} > t] &= P[W^{(2)} > t] + (1-\lambda) P[X_2 > t] + \int_{0}^{t} f_{W^{(2)}}(s) P[X_2 > t-s] \,ds \\
        &= P[W^{(2)} > t] + (1-\lambda) \alpha_2 e^{S_2 t} \e + \int_{0}^{t} ((e^*_1 \otimes \lambda \beta, 0) e^{T_M s} (-T_M) v_2) (\alpha_2 e^{S_2 (t-s)} \e) \,ds % \\
        %&= P[W^{(2)} > t] + (1-\lambda) \alpha_2 e^{S_2 t} \e - (e^*_1 \otimes \lambda \beta, 0) \left(\int_{0}^{t} e^{T_M s} T_M v_2 \alpha_2 e^{S_2 (t-s)} \,ds\right) \e
    \end{align*}
    This gives the desired formula using Lemma 1 from \citep{vanhoudt_nudge}.
\Halmos \end{proof}

We also developed a numerical method for the response time distribution of a type-1 job under Nudge*$(M)$ by making use of a Markov Modulated Fluid  Queue (MMFQ) \citep{dzial1,vanhoudt30,soares3,horvath5}. We defer this to Section \ref{app:t1} as the full description of the method is quite lengthy. A numerical method to directly compute the mean response time is presented in Section \ref{app:mean}, which has a better time complexity than using the results from this section and Section \ref{app:t1}.

\begin{figure*}[t!]
\begin{subfigure}[t]{.48\textwidth}
  \centering
  \includegraphics[width=1\linewidth]{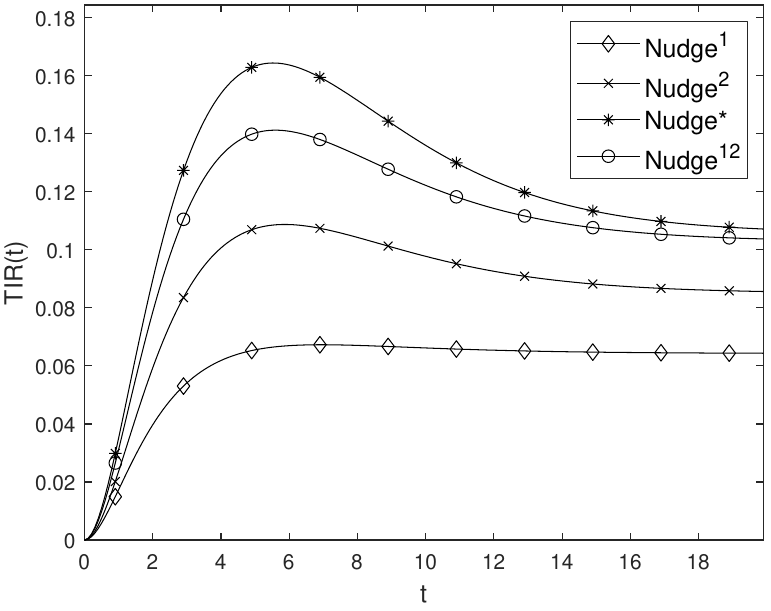}
	\caption{Exponential type-1 and type-2 job sizes, $K = L= M_{opt} = 4$, Nudge$^{12}$ with $(K, L) = (2, 3)$}
	\label{fig:TIRexp}
\end{subfigure}
\begin{subfigure}[t]{.48\textwidth}
  \centering
  \includegraphics[width=1\linewidth]{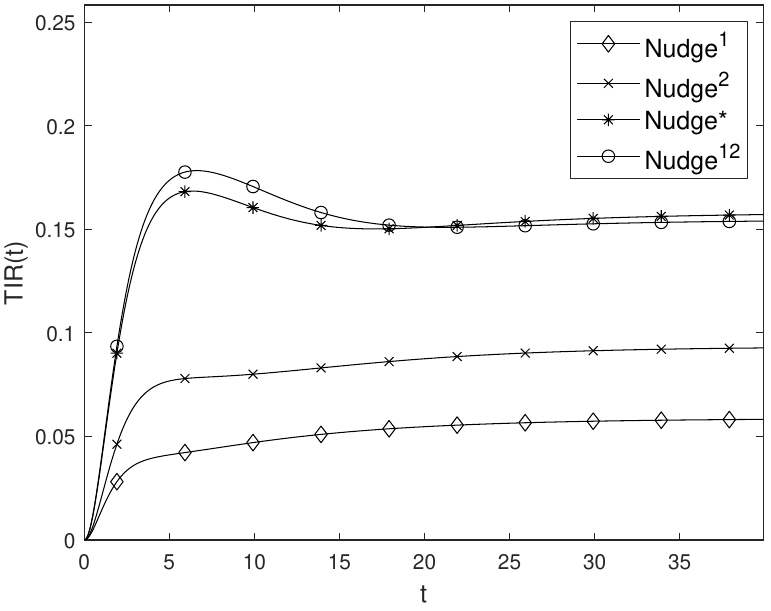}
	\caption{Exponential type-1 jobs, Hyperexponential ($SCV = 2$, $f = 1/2$) type-2 jobs, $K = L = M_{opt} = 7$, Nudge$^{12}$ with $(K, L) = (4, 6)$}
	\label{fig:TIRerlang8}
\end{subfigure}
\caption{The tail improvement ratio of Nudge$^1(K)$, Nudge$^2(L)$, Nudge*$(M)$, and Nudge$^{12}(K, L)$ (with parameters that maximize the ATIR) over FCFS with $E[X_2]/E[X_1]=3$, $p=2/3$, $\lambda=0.7$. Nudge*$(M)$ outperforms the other Nudge algorithms asymptotically, but the difference with Nudge$^{12}(K,L)$ is often small. The second plot shows that Nudge*$(M)$ does not always stochastically improve upon Nudge$^{12}(K, L)$.
}
\label{fig:TIR}
\end{figure*}

We end this section by comparing the tail improvement ratio in $t$ (TIR($t$)) over FCFS, defined as 
$1-P[R_A > t]/P[R_{FCFS}>t]$ for scheduling algorithm $A$, of Nudge$^1(K)$, Nudge$^2(L)$, Nudge*$(M)$ and Nudge$^{12}(K,L)$ in Figure \ref{fig:TIR}. In all cases
the parameters $K,L$ and $M$ were set such that the ATIR was maximized.
To compute this ratio for Nudge$^{12}(K,L)$ we also developed a computational approach similar
to the one for Nudge*$(M)$ presented in this section %Section \ref{sec:t2}
and in Section \ref{app:t1}. 
The numerical results show that the asymptotic improvement often corresponds to a stochastic improvement,  
though examples as in \cite[Figure 3]{vanhoudt_nudge} can be constructed where this is not the case.

\section{Conclusions and future work}\label{sec:conc}
The Nudge*$(M)$ scheduling algorithm was introduced in this paper in a setting where jobs are partitioned into two
types and the scheduler only knows the type of a job (but not its size) upon arrival. 

A simple explicit formula was presented for the optimal parameter $M$, denoted as $M_{opt}$, 
which is expressed in terms of the Laplace transforms of the type-1 and type-2 job size distributions.
The Nudge*$(M_{opt})$ algorithm was shown to be strongly tail optimal within a large family $\mathcal{F}$ of Nudge scheduling algorithms. Moreover it was shown that in the heavy traffic setting the prefactor of Nudge*$(M_{opt})$ and $\gamma$-Boost coincide, meaning  the additional reduction in the tail offered by using arrival time information vanishes as the load tends to one. 
The heavy traffic analysis also yielded an approximation $M_{heavy}$ for $M_{opt}$ which may be easier to estimate in practice
as $M_{heavy}$ only relies on the
arrival rate, the mean size of type-1 and type-2 jobs and the second moment of the overall job size.
Finally, a computational method to evaluate the response time distribution of Nudge*$(M)$ (and its mean)  was developed
using Markov modulated fluid queues. 

Interesting future directions for research include the case with a more general arrival process or when the jobs are partitioned into more than $2$ types. As the Nudge*$(M)$ algorithm can be regarded as a scheduling algorithm with a boosted arrival order, generalizing the Nudge*($M$) algorithm to a setting with more than two types is not hard. However the approach used to prove optimality within
a large class of Nudge algorithms does not appear to generalize as the conditions required on the $n$ function
in order to have a feasible algorithm are much more involved as soon as there are three types of jobs.

%\section*{Acknowledgement}
\ACKNOWLEDGMENT{
The authors thank Ziv Scully for his suggestion that Nudge*$(M)$ could be optimal
among scheduling algorithms that only use job type information and their order of arrival.
This work was supported by the FWO project G0A9823N.}

%\bibliographystyle{informs2014}
%\bibliography{thesis}

\section*{Author Biographies}

{\bf Nils Charlet} received his M.Sc. degree in Data Science and Artificial Intelligence from the University of Antwerp (Belgium) in 2023, where he is currently pursuing his PhD degree with the department of Computer Science, under the supervision of Benny Van Houdt. His main research interests include the stochastic modeling and analysis of queuing systems.

{\bf Benny Van Houdt} is a professor at the Mathematics and Computer Science Department at the University of Antwerp (Belgium). His main research interest goes to the performance evaluation and stochastic modeling of computer systems and communication networks. He has received several awards, including best paper awards at ACM Sigmetrics and IFIP Performance.

%\begin{APPENDIX}
\ECSwitch

\ECHead{E-companion}

This E-companion contains the proofs of Lemma 1, 2, 3 and 4. In addition a numerical method to compute the mean response time is presented in Section EC.5, while  Section EC.6 indicates how to compute the response time distribution for a type-1 job.

\section{Proof of Lemma 1}\label{app:CW1}
In this Section we present a detailed proof for Lemma 1.

{\it Proof:} We make use of a coupled FCFS queue. This queue is identical to our Nudge*$(M)$ queue, except that
jobs are served in FCFS order. In other words the arrival times are the same in both queues and the arriving
jobs have exactly the same type and size in both systems. 

Let $\tilde Z(s)$ be the Laplace transform of the workload of the coupled FCFS queue (which is identical
to the workload in the Nudge*$(M)$ queue). Note that this is also the FCFS workload seen by a type-1 arrival
due to the PASTA property. Let $p_+$ be the probability that an arriving type-$1$
job sees at least $M$ jobs waiting in the FCFS queue. Let $Z_+$ (and $\tilde Z_+(s)$)   be the
FCFS workload (and its transform) seen by an arriving type-$1$ job when there are at least $M$ waiting jobs in the FCFS queue.
If there are less than $M$ jobs waiting in the FCFS queue, we denote the workload and its Laplace transform
as $Z_-$ and $\tilde Z_-(s)$, respectively. Hence,
\[ P[Z > t] = p_+ P[Z_+ >t] + (1-p_+)P[Z_- > t].\]
and
\[ \tilde Z(s) = p_+ \tilde Z_+(s) + (1-p_+)\tilde Z_-(s).\]
The workload in the event that fewer than $M$ jobs are waiting in the FCFS queue decays as 
$\min(\theta_1,\theta_2)$ as it is the superposition of at most $M$ jobs. 
This implies that
\[ c_Z = \lim_{t \rightarrow \infty} e^{\theta_Z t} P[Z > t] = \lim_{t \rightarrow \infty} e^{\theta_Z t}  p_+ P[Z_+ >t],\]
as $\theta_Z < \min(\theta_1,\theta_2)$.
We further note that $\tilde Z_+(s)$ can be written as
\[ \tilde Z_+(s) = \tilde Z_+^{(r)}(s)  \tilde S(s)^{M},\]
where $\tilde Z_+^{(r)}(s)$ is the Laplace transform of the workload seen by a type $1$ job that
encounters at least $M$ jobs in the coupled FCFS queue, except for the work belonging to the last
$M$ waiting jobs (in the FCFS queue).
Let $W_M^{(1)}$ (and $\tilde W_M^{(1)}(s)$) be the waiting time (and its transform) of a type $1$ job in the Nudge*$(M)$
queue. Similarly we use the notations $W_{M+}^{(1)}$, $W_{M-}^{(1)}$, $\tilde W_{M+}^{(1)}(s)$ and $\tilde W_{M-}^{(1)}(s)$ for the 
waiting times in the Nudge*$(M)$ queue of a type $1$ job depending on whether this job encounters at least
$M$ jobs in the coupled FCFS queue. Hence,
\[ P[W_M^{(1)} > t] = p_+ P[W_{M+}^{(1)} >t] + (1-p_+)P[W_{M-}^{(1)} > t].\]
and
\[ \tilde W_M^{(1)}(s) = p_+ \tilde W_{M+}^{(1)}(s) + (1-p_+)\tilde W_{M-}^{(1)}(s).\]
Clearly $W_{M-}^{(1)}$ also decays as $\min(\theta_1,\theta_2)$ and we therefore have
\[ c_{W^{(1)}}(M) = \lim_{t \rightarrow \infty} e^{\theta_Z t} P[W_M^{(1)} > t] = \lim_{t \rightarrow \infty} e^{\theta_Z t}  p_+ P[W_{M+}^{(1)} >t].\]

The key observation now is that when a type-$1$ job arrives and there are at least $M$ jobs in the
coupled FCFS queue, then we can determine the waiting time of that type-$1$ job by looking at the 
types of the last $M$ jobs in the coupled FCFS queue. All the type-$2$
jobs among these $M$ jobs are swapped with the arriving type-$1$ job, because any type-2 job waiting in the coupled FCFS
queue is also waiting in the Nudge*$(M)$ queue. In other words, the waiting time of the
type $1$ job can be expressed as $Z^{(r)}_+$, which is the workload without these $M$ last jobs, 
plus the work associated to all the type $1$ jobs among the last $M$ jobs in the
coupled FCFS queue.
Hence, with probability $\binom{M}{k} (1-p)^k p^{M-k}$ we need to add the work of $M-k$ type-$1$ jobs, which yields 
(using the final value theorem)
\begin{align}
    c_{W^{(1)}}(M) &= \frac{c_Z}{\tilde S(-\theta_Z)^M} \sum_{k=0}^M \binom{M}{k} (1-p)^k p^{M-k}  \tilde S_1(-\theta_Z)^{M-k} \nonumber \\
    &= c_Z \sum_{k=0}^M \binom{M}{k} \left(\frac{(1-p)}{\tilde S(-\theta_Z)}\right)^k 
    \left( \frac{p\tilde S_1(-\theta_Z)}{\tilde S(-\theta_Z)}\right)^{M-k} \nonumber \\
    &= c_Z (w_1+w)^M.
   \end{align}
To see that $w+w_1 < 1$, note that $\tilde{S}(-\theta_Z), \tilde{S_1}(-\theta_Z),\tilde{S_2}(-\theta_Z) > 1$ holds as $\theta_Z > 0$. Combining this with $p \tilde{S_1}(-\theta_Z) + (1-p) \tilde{S_2}(-\theta_Z) = \tilde{S}(-\theta_Z)$ implies that $w$, $w_1$, and $w + w_1$ are in $(0, 1)$. 
\Halmos 

\section{Proof of Lemma 2}\label{app:CW2}

In this Section we present a detailed proof for Lemma 2.

{\it Proof:} When determining the prefactor $c_{W^{(2)}}(M)$ of the waiting time of a type-2 job, 
we may assume that the next $M$ arrivals occur
while the tagged type-2 job is still waiting. 
This can be seen as follows. Pick any $s > 0$. We can distinguish between type-2
jobs that see $w > s$ work upon arrival and type-2 jobs that see $w \leq s$ work.
For the latter jobs their waiting time decays as $\min(\theta_1,\theta_2) > \theta_Z$.
Therefore, 
\[\lim_{t \rightarrow \infty} e^{\theta_Zt} P[W_M^{(2)} > t | Z_2 \leq s]=0,\]
where $Z_2$ is the workload seen by a type-2 arrival.
Further, for any $\epsilon > 0$ we can pick $s$ such that the probability of
having $M$ or more arrivals in $(0,s)$ is at least $1-\epsilon$. Hence, we may assume
when computing $\lim_{t \rightarrow \infty} e^{\theta_Zt} P[W_M^{(2)} > t]$
that at least $M$ arrivals occur while the type-2 job is waiting.

With probability $\binom{M}{k} (1-p)^{M-k} p^{k}$
there are $k$ type-1 jobs among the next $M$ arrivals and these $k$ jobs pass the type-2 job. 
The waiting time of the type-2 job therefore
equals the FCFS workload plus the work associated with $k$ type-1 jobs. 
Hence, by the final value theorem we find
\begin{align}
    c_{W^{(2)}}(M) &= c_Z \sum_{k=0}^M \binom{M}{k} (1-p)^{M-k} p^k  \tilde S_1(-\theta_Z)^{k} = c_Z (w_1+w)^M
    \tilde S(-\theta_Z)^M.
    \end{align}
     \Halmos

\section{Proof of Lemma 3}\label{app:CWgen}

In this Section we present a detailed proof for Lemma 3.

{\it Proof:} By repeating the arguments in the proof of Theorem 1, we may assume that a type-2 job
that arrived among the last $M$ arrivals seen by a tagged type-1 arrival is still present in the queue
when studying the prefactor $c_{W^{(1)}}$. A tagged type-1 job sees the string $s$ as the string of the last 
$M$ arrivals with probability $(1-p)^{t(s)} p^{M-t(s)}$.  In this case 
the type-1 job passes $n(s)$ type-2 jobs, such that its waiting time equals the FCFS workload minus
the work of the last $M$ arrivals plus the work of $M-t(s)$ type-1 jobs plus the work of $t(s)-n(s)$
type-2 jobs. Summing over $s \in \{1,2\}^M$ yields
\begin{align*}
    c_{W^{(1)}} = \frac{c_Z}{\tilde S(-\theta_Z)^M} 
    \sum_{s \in \{1,2\}^M} (1-p)^{t(s)} p^{M-t(s)} \tilde S_1(-\theta_Z)^{M-t(s)} \tilde S_2(-\theta_Z)^{t(s)-n(s)},
\end{align*}
where $c_Z/\tilde S(-\theta_Z)^M$ covers the work without the last $M$ arrivals, the factor $\tilde S_1(-\theta_Z)^{M-t(s)}$ adds the work of the type-1 jobs in the last $M$ arrivals and $\tilde S_2(-\theta_Z)^{t(s)-n(s)}$ the
work of the type-2 jobs in the last $M$ arrivals that are not passed.

As in the proof of Theorem 3 we may assume that the next $M$ arrivals occur while
a tagged type-2 job is still waiting when considering $c_{W^{(2)}}$. 
The expression for $c_{W^{(2)}}$ is somewhat more involved to determine as a tagged type-2 job can be 
passed by multiple type-1 jobs that do not necessarily see the same string $s$, but do see overlapping strings. 
We therefore look at
a sequence of $2M$ arrivals where the tagged type-2 job is the $M+1$-th oldest of the $2M$ arrivals. Further, as the arrivals in position $M+2$ to $2M$ in $s$ already occurred when the
tagged type-2 job arrives, we need to condition the FCFS workload on the types $s_{M+2} \ldots s_{2M}$ 
of the last $M-1$ arrivals before the tagged type-2 job.
Note that the
type-1 jobs that may pass the tagged type-2 job are among the most recent $M$ arrivals
(with types $s_1 \ldots s_M$). A type-1 job in position $k \leq M$ passes the tagged type-2
job in position $M+1$ if $n(s_{k+1}\ldots s_{k+M})>t(s_{k+1}\ldots s_M))$, yielding:
\begin{align*}
    c_{W^{(2)}} = c_Z 
    \sum_{\substack{s \in \{1,2\}^{2M} \\ s_{M+1}=2}} & \frac{(1-p)^{t(s)} p^{2M-t(s)}}{(1-p)} 
    \frac{\tilde S_1(-\theta_Z)^{M-1-t(s_{M+2}\dots s_{2M})}\tilde S_2(-\theta_Z)^{t(s_{M+2}\dots s_{2M})}}{\tilde S(-\theta_Z)^{M-1}} \\
    &\hspace*{1cm}  \cdot \prod_{k=1}^{M} \tilde S_1(-\theta_Z)^{1(s_k = 1 \wedge n(s_{k+1}\ldots s_{k+M})>t(s_{k+1}\ldots s_M))},
\end{align*}
where the second fraction in the sum is due to the conditioning on the FCFS workload. \Halmos

\section{Proof of Lemma 4}\label{app:ns}
In this Section we present a detailed proof for Lemma 4.

{\it Proof:}
 If $n(s)$ is increased by one such that the conditions (C1) and (C2) on $n$ remain valid, 
    one finds that $c_{W^{(1)}}$ changes by
    \begin{align*}
    \Delta c_{W^{(1)}} &= 
    \frac{c_Z}{\tilde S(-\theta_Z)^M} 
    (1-p)^{t(s)} p^{M-t(s)} \tilde S_1(-\theta_Z)^{M-t(s)} \tilde S_2(-\theta_Z)^{t(s)-(n(s)+1)} (1-\tilde S_2(-\theta_Z)),
        \end{align*}
        as when a type-1 jobs sees the string $s$ upon arrival it now passes $n(s)+1$ type-2 jobs as
        opposed to $n(s)$, which results in the correction of a single term in $c_{W^{(1)}}$.

     We now focus on the change to $c_{W^{(2)}}$. By increasing $n(s)$ by one, the power of 
     $\tilde S_1(-\theta_Z)$ needs to be increased by one for some of the length $2M$ strings. The strings $s'=s'_1 \ldots s'_{2M}$ that require a change are such that 
     \begin{itemize}
         \item $s=s'_{k+1} \ldots s'_{k+M}$ for some $k \in \{1,\ldots,M\}$,
         \item the tagged type-2 job is in position $M+1$ and a type-1 job is in position $k$, that is, $s'_k = 1$,
         \item  the tagged type-2 job is the $n(s)+1$-st type-2 job in the string $s$.
     \end{itemize} 
     Notice that in such case the type-1 job that is in position $k$ of $s'$ did not pass the
     tagged type-2 job in position $M+1$ of $s'$ before increasing $n(s)$ by one, but it does pass this tagged job due to the increase of $n(s)$ as the tagged job is the $n(s)+1$-st type-2 job in the string $s$ seen 
     upon arrival by the type-1
     job in position $k$ of $s'$.
    The value of $k$ is fully determined by noting that if the $n(s)+1$-st two is in position $k'\geq n(s)+1$ of $s$, then $k=M+1-k'$.
     There can clearly be several strings $s'$ of length $2M$ for which the above listed $3$ requirements hold.
     The change to $c_{W^{(2)}}$ will be $(\tilde S_1(-\theta_Z)-1) > 0$ times the sum of all the terms associated
     to these strings $s'$. For these strings $s'$ the entries in position $k$ to $k+M$ contribute a factor equal to
     \begin{align}
         p \frac{(1-p)^{t(s)} p^{M-t(s)}}{1-p} &\tilde S_1(-\theta_Z)^{M-k-n(s)} 
         \frac{\tilde S_1(-\theta_Z)^{M-k'-(t(s)-(n(s)+1))}\tilde S_2(-\theta_Z)^{t(s)-(n(s)+1)}}{\tilde S(-\theta_Z)^{M-k'}}
     = \nonumber\\
         &p \frac{(1-p)^{t(s)} p^{M-t(s)}}{1-p} \frac{\tilde S_1(-\theta_Z)^{M-t(s)}\tilde S_2(-\theta_Z)^{t(s)-(n(s)+1)}}{\tilde S(-\theta_Z)^{M-k'}} \label{eq:contr_ktok+M}
     \end{align} 
    as there are $M-k-n(s)$ type-1 jobs in $s'_{k+1} \ldots s'_M$ which all pass the type-2 job in position 
    $M+1$ of $s'$ as increasing $n(s)$ would otherwise violate condition (C2). 
    The second fraction takes care of the conditioning on the types of the jobs after the tagged
    type-2 job up to position $k+M$. There are $M-k'$ such jobs and $t(s)-(n(s)+1)$ of them are type-2.
    We divide by $(1-p)$ as we
    condition on having a type-2 job in position $M+1$. 
    The entries in position $k+M+1$ to $2M$ in $s'$ can be of any type
     and do not affect the value of \eqref{eq:contr_ktok+M}. 
     A type-1 job contributes a factor $p\tilde S_1(-\theta_Z)/\tilde S(-\theta_Z)$ and a type-2 job a factor $(1-p)\tilde S_2(-\theta_Z)/\tilde S(-\theta_Z)$.
     They can therefore be safely ignored
     as $p\tilde S_1(-\theta_Z)+(1-p)\tilde S_2(-\theta_Z) = \tilde S(-\theta_Z)$, meaning 
     $(p\tilde S_1(-\theta_Z)+(1-p)\tilde S_2(-\theta_Z))^{M-k}/\tilde S(-\theta_Z)^{M-k}=1$.
     
   Finally, if we look at the factor contributed by one of the first $k-1$ entries in $s'$, we note that any type-1
   job among the first $k-1$ cannot pass the tagged type-2 job in position $M+1$ as the type-1 job in position
   $k$ could not pass the tagged job either (before increasing $n(s)$ by one) and condition (C2) holds
   on $n$. This means that the jobs in the
   first $k-1$ positions either contribute a factor $1-p$ or $p$. Further \eqref{eq:contr_ktok+M} 
   does not depend on $s'_1$ to $s_{k-1}'$.
   Summing over the different possibilities for the
   types of these $k-1$ jobs simply yields a factor $1$ as $((1-p)+p)^{k-1}=1$.
     Combining the above yields the following expression for  $\Delta c_{W^{(2)}}$:
    \begin{align*}
    \Delta c_{W^{(2)}} & = 
    c_Z  p \frac{(1-p)^{t(s)} p^{M-t(s)}}{1-p} 
    \frac{\tilde S_1(-\theta_Z)^{M-t(s)}\tilde S_2(-\theta_Z)^{t(s)-(n(s)+1)}}{\tilde S(-\theta_Z)^{M-k'}}
     (\tilde S_1(-\theta_Z)-1), 
    %\\ &\leq c_Z  \frac{p}{1-p} (1-p)^{t(s)} p^{M-t(s)} \tilde S_1(-\theta_Z)^{M-n(s)-1}  (\tilde S_1(-\theta_Z)-1),
    \end{align*}
    The ATIR improves when increasing $n(s)$ by one if and only if
    \[p \tilde S_1(-\theta_Z)  \Delta c_{W^{(1)}} + (1-p) \tilde S_2(-\theta_Z) \Delta c_{W^{(2)}}< 0.\]
    Using the expression for $\Delta c_{W^{(1)}}$ and $\Delta c_{W^{(2)}}$, we
    have
    \begin{align} \label{eq:deltaeq}
        \frac{c_Z}{\tilde S(-\theta_Z)^{M-k'}} (1-p)^{t(s)} &p^{M-t(s)} p\tilde S_1(-\theta_Z)^{M-t(s)}  \tilde S_2(-\theta_Z)^{t(s)-(n(s)+1)} \nonumber  \\
    & \left(
    \frac{1}{\tilde S(-\theta_Z)^{k'}} \tilde S_1(-\theta_Z)
    (1-\tilde S_2(-\theta_Z))+ \tilde S_2(-\theta_Z) (\tilde S_1(-\theta_Z)-1) 
    \right) < 0,
    \end{align}
    which can be restated as the ATIR improves if and only if
    \begin{align}
    \frac{1}{\tilde S(-\theta_Z)^{k'}}\frac{\tilde S_1(-\theta_Z)(\tilde S_2(-\theta_Z)-1)}{\tilde S_2(-\theta_Z)(\tilde S_1(-\theta_Z)-1)} > 1,
    \end{align}
    where $k'$ is the position where the $n(s)+1$-st two appears in $s$.
    By definition of $M_{opt}$ the above condition holds if and only if $k' \leq M_{opt}$.
\Halmos

\section{Mean response time of Nudge*$(M)$}\label{app:mean}

In this Section we present a numerical method to compute the mean response time of Nudge*$(M)$ in $O(M^3)$ time, for phase-type job sizes.
The idea is to compute the mean response time based on the mean response time of the FCFS scheduling algorithm
$E[R_{FCFS}]$
and the mean number of jobs that pass a random type-$2$ job, denoted as $E[X_{swap}]$, using the following theorem: 

\begin{theorem}
    The mean response time of Nudge*$(M)$ can be expressed as
    $$
    E[R_{Nudge-M}] = E[R_{FCFS}] + (1-p) E[X_{swap}] (E[X_1] - E[X_2])
    $$
    where $E[R_{FCFS}] = 1 + \lambda \beta T^{-2} \e$ is the mean response time of FCFS.
\end{theorem}
\begin{proof}{Proof:}
    An arriving job is type-2 with probability $1-p$, and is swapped on average $E[X_{swap}]$ times. This means per arriving job of any type, the average number of swaps is $(1-p) E[X_{swap}]$, and as each swap moves a type-1 job in front of a type-2 job, the response time of the type-1 job decreases by $E[X_2]$ while that of the type-2 job increases by $E[X_1]$.
\Halmos \end{proof}

Using the matrix exponential form of both $Z$ and $X_1$
and the well-known fact that the generating function $\hat{A}_X(z)$ of the number of arrivals of a Poisson
process with rate $\lambda$ during a time
interval with duration $X$  is given by $\tilde{X}(\lambda (1-z))$, where $\tilde X$ is the Laplace transform of $X$, 
we have:
\begin{align*}
    P[A_Z = k] &= 
        \lambda^{k+1} \beta (\lambda I - T)^{-1-k} \e + 1[k=0] (1-\lambda),
\end{align*}    
and
\begin{align*}    
    P[A_{X_1} = k] &= \lambda^{k} \alpha_1 (\lambda I - S_1)^{-1-k} (-S_1) \e.
\end{align*}
We now introduce the probabilities $q_{i,j}$ with $0 \leq i \leq j \leq M$. 
The value of $q_{i,j}$ is given by the probability that there are at least $j$ arrivals during the waiting time of
a random type-$2$ job and exactly $j$ arrivals occurred during the 
initial wait $Z$ plus the service time all the type-$1$ jobs among the first $i$ arrivals that pass the
random type-$2$ job. Notice that by definition $q_{j,j}$ contains the probability that exactly $j$
jobs arrive during the waiting time of a type-$2$ job, for $j < M$. 
Hence, for $1 \leq k \leq M$ the probability that the $k$-th arrival after the tagged job arrives during the tagged job's waiting time is given by $\sum_{\ell=k}^M q_{\ell,\ell}$, and with probability $p$ it passes the tagged job. This gives the following  formula for $E[X_{swap}]$:
\begin{align}
    E[X_{swap}] &= p \sum_{k=1}^M \sum_{\ell=k}^M q_{\ell,\ell} = p \sum_{k=1}^M k q_{k,k}
\end{align}
In fact, by limiting the first sum to $n < M$, we can also get $E[X_{swap}]$ for Nudge*$(n)$ for $n < M$. 

We proceed by showing how to compute the $q_{i,j}$ values recursively. We initialize $q$ as $q_{0,j} = P[A_Z = j]$ for $j < M$ and $q_{0,M} = P[A_Z \geq M]$. Then for $i = 1 \ldots M$, $j = i \ldots M$, calculate $q_{i,j}$ as
\begin{align}
    q_{i,j} = (1-p) q_{i-1, j} + p \left\{ 
    \begin{array}{ll}
        \sum_{k=0}^{j-i} P[A_{X_1} = k] q_{i-1, j-k} & \quad j < M, \\
        \sum_{k=0}^{j-i} P[A_{X_1} \geq k] q_{i-1, j-k} & \quad j = M.
    \end{array}
    \right.
\end{align}
As the $i$-th arrival during the waiting time adds $k$ more arrivals if it is a type-$1$ job, the job size
of which contains $k$ new arrivals. This recursion can be implemented in $O(M^3)$ time and $O(M)$ space.
It is also possible to generalize this recursive approach in order to compute the distribution of
the number of swaps $X_{swap}$.

In Figure \ref{fig:ERIR} we plot the mean response time improvement ratio (MTIR) of Nudge*$(M)$ and a non-preemptive priority
scheduler \citep{takacs_priority} over FCFS. The MTIR is defined as $1-E[R_A]/E[R_{FCFS}]$ for a scheduling algorithm $A$.  The non-preemptive priority scheduler serves as an upper bound here as it minimizes the mean response time
when the server only knows job types (and $E[X_1] < E[X_2]$). 
We first note that Nudge*$(M)$ with $M=M_{opt}$ achieves a significant gain in the mean response time over FCFS. It also captures
a substantial part of the gain achieved by the priority scheduler. The priority scheduler however has a poor tail performance
as its tail only decays exponentially fast for sufficiently small loads \citep{abate_priority}. We further note that Nudge*$(20)$ acts as a priority scheduler for small loads, but it does not have
the same guarantees in terms of the tail behavior of the response time as Nudge*$(M)$ with $M=M_{opt}$.

\begin{figure*}[t!]
\begin{subfigure}[t]{.48\textwidth}
  \centering
  \includegraphics[width=\linewidth]{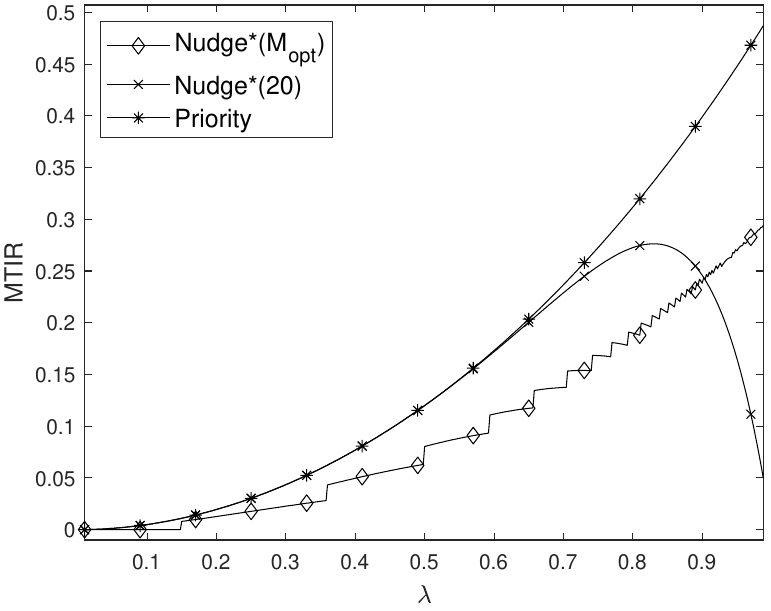}
  \caption{Exponential type-1 and type-2 job sizes}
\end{subfigure}
\begin{subfigure}[t]{.48\textwidth}
  \centering
  \includegraphics[width=\linewidth]{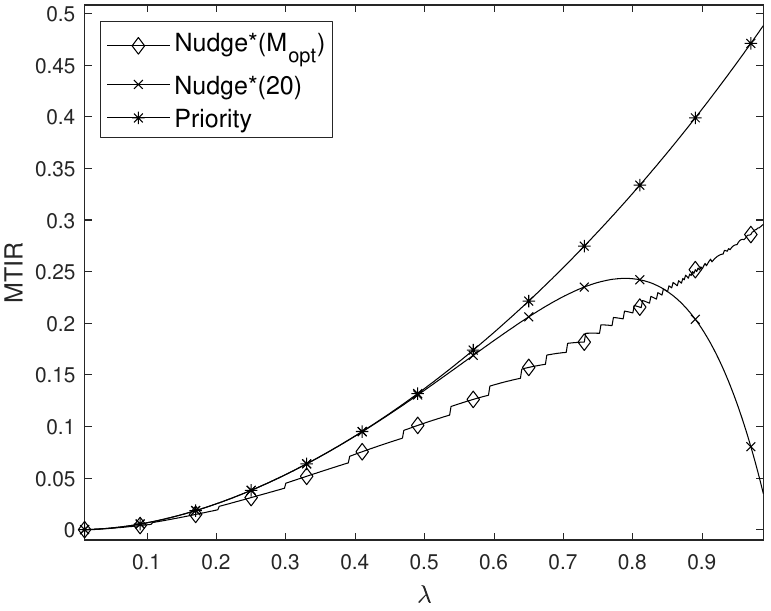}
  \caption{Exponential type-1 jobs, hyperexponential ($SCV = 2$, $f = 1/2$) type-2 jobs}
\end{subfigure}
\caption{The mean response time improvement ratio (MTIR) of Nudge*$(M)$ with $M = M_{opt}$, Nudge*$(M)$ with $M = 20$, and a priority queue over FCFS with $p = 2/3$, and $E[X_2]/E[X_1] = 4$. For low loads, Nudge*$(M)$ with a high parameter coincides with the priority queue. This plot also shows how Nudge*$(M)$ improves the mean response time without making the tail worse than FCFS.}
\label{fig:ERIR}
\end{figure*}

\section{Response time distribution of type-1 jobs}\label{app:t1}

This Section focuses on the computation of the response time distribution of a type-1 job under Nudge*$(M)$, when the job sizes follow a phase-type distribution. In \citep{vanhoudt_nudge}
the response time of a type-1 job under Nudge$^1(K)$ could be computed by creating a coupling with the FCFS queue. To illustrate that
this approach can no longer be used, consider the following example with $M=2$.
A type-2 job is in service and the queue is empty at time $0$. Next a type-2 job arrives followed by a type-1 job.
In Nudge*$(M)$ the type-1 is swapped with the type-2 job. Suppose the type-2 job now completes service.
In Nudge*$(M)$ the type-1 job enters the server, while in the coupled FCFS system the type-2 job starts service.
Assume the next arrival is another type-1 job, call it job J. In Nudge*$(M)$ there are 2 cases:
(1) the type-1 job was small and already left the system and the type-2 is in service, in this case job J is not swapped,
(2) the type-1 job is still in service, in this case job J is swapped.
When we look at the FCFS system when job J arrives, it simply has a type-2 job in service and a type-1 waiting
(if the type-2 job is large enough), so it is unclear
how one could determine which of the two cases occurred under Nudge*$(M)$ simply by looking at the FCFS state.
We therefore need to develop a radically different approach compared to one taken in \citep{vanhoudt_nudge}.

We proceed by showing how a Markov Modulated Fluid Queue (MMFQ) with jumps \citep{dzial1,vanhoudt30} can be used to model the waiting time distribution of a type-1 job. A fluid queue {\it without jumps} consists of a non-negative fluid that evolves over time depending on a background CTMC's state \citep{soares3,horvath5}. This background process' states are partitioned into two sets when the fluid is non-zero. The first set $\mathcal{S}^-$ makes the fluid decrease at rate $1$, while the second set $\mathcal{S}^+$ makes the fluid increase at rate $1$. When the fluid is exactly zero, the set $\mathcal{S}^-$ is replaced by another set $\mathcal{S}^0$, in which the fluid remains at zero until a transition to $\mathcal{S}^+$ is made. The difference with a MMFQ with jumps is that the time intervals in which the fluid
increases at rate $1$ are replaced by instantaneous upward jumps. In other words the MMFQ with jumps is identical
to the one without jumps, except that we censor out the time intervals in which the fluid increases, that is,
when the background process's state is in $\mathcal{S}^+$. This implies that we can compute the stationary
distribution of the fluid of the MMFQ with jumps from the stationary distribution of the MMFQ without jumps.

The MMFQ without jumps used in this section is defined using three matrices: a rate matrix $T$ that captures the evolution of
the background process when the fluid is nonzero, a second rate matrix $T^*$ for the background process when the fluid is exactly zero, and a transition probability matrix $P$ which allows for the background process to make a transition when the fluid reaches zero. The matrix $P$ is of size $|\mathcal{S}^-| \times (|\mathcal{S}^0| + |\mathcal{S}^+|)$ as the fluid can reach zero only from states where the fluid decreases. The matrix $T^*$ is of size $|\mathcal{S}^0| \times (|\mathcal{S}^0| + |\mathcal{S}^+|)$ as the fluid can only stay at zero from states in $\mathcal{S}^0$. The matrix $T$ is clearly
a square matrix of size $|\mathcal{S}^-| + |\mathcal{S}^+|$. These matrices are partitioned as
\begin{align}
    T = \begin{bmatrix}
        T_{--} & T_{-+} \\ T_{+-} & T_{++}
    \end{bmatrix}, \hspace*{5mm} 
    T^* = \begin{bmatrix}
        T^*_{00} & T^*_{0+} 
    \end{bmatrix}, \hspace*{5mm}  \mbox{and} \hspace*{5mm} P = \begin{bmatrix}
        P_{-0} & P_{-+} 
    \end{bmatrix}
\end{align}
in the obvious manner. Notice that if $P_{-+}$ is nonzero, the fluid immediately starts to increase again
when hitting the zero boundary in some cases. 
The idea is to let the fluid level of the MMFQ {\it with jumps} represent the waiting time of a virtual type-1 arrival. 
We first illustrate this with two simpler scheduling policies before moving on to Nudge*$(M)$. 

Before doing this, we briefly explain how the stationary distribution of the fluid with jumps can be found 
using these matrices (see \citep{horvath5} for more details).
Let $\Psi$ be the smallest non-negative solution to the algebraic Riccati equation 
\begin{align}\label{eq:Ricat}
    T_{+-} + \Psi T_{--} + T_{++} \Psi + \Psi T_{-+} \Psi = 0,
\end{align}
which can be solved efficiently using the Structure-preserving Doubling Algorithm (SDA) \citep{guo3}, the 
computation time of which
can be slightly improved upon by the ADDA algorithm \citep{wang4}. Alternatively, one can also construct
a Quasi-Birth-Death Markov chain to compute $\Psi$ as explained in \citep{ramaswami5}.
Define two matrices
\begin{align}
        \tilde{P} &= P_{-+} - P_{-0} (T^*_{00})^{-1} T^*_{0+}, \\
        K &= T_{++} + \Psi T_{-+},
\end{align} and $\pi_+$ as a non-zero solution of $\pi_+ \Psi \tilde{P} = \pi_+$. The normalizing constant $\eta$ is then calculated as:
\begin{align}
    \eta = -\pi_+ (K^{-1} \Psi \e + \Psi P_{-0} (T^*_{00})^{-1} \e).
\end{align}
Normalize $\pi_+$ so that $\eta$ is equal to 1.
Finally, the probability mass $c_0$ that the fluid is exactly 0 (which is the probability that the queue is empty), the density of the fluid (i.e. the density of the waiting time distribution for a type-1 job), and its complementary distribution function are given by the following formulas:
\begin{align}
    c_0 &= -\pi_+ \Psi P_{-0} (T^*_{00})^{-1} \e \\
    f_{W^{(1)}}(t) &= \pi_+ e^{Kt} \Psi \e \\
    P[W^{(1)} > t] &= \int_{t}^{\infty} f_{W^{(1)}}(x) \,dx = \pi_+ e^{Kt} (-K)^{-1} \Psi \e  \label{eq:fluidqueue}
\end{align}

\subsection{FCFS scheduling:} 
In case of FCFS scheduling the waiting time of a virtual type-1 arrival is simply the workload in the queue,
so the fluid level should reflect the workload.
Thus whenever a job arrives (of any type), its workload is added to the fluid of the MMFQ via the states 
$\mathcal{S}^+$. This implies that $T_{-+} = \lambda \alpha$, $T_{++}=S$ and $T_{+-}=s^*$.
In between arrivals, the fluid decreases at rate $1$ as the FCFS scheduler is work-conserving. The set $\mathcal{S}^-$ contains only one state and is equal to $\mathcal{S}^0$, so $T_{--}=-\lambda$. If the fluid is at 0, it remains at 0 until the next arrival occurs. 
Hence, $P_{-+}=\vecn$, $P_{-0}=1$, $T^*_{00} = -\lambda$  and $T^*_{0+} = \lambda \alpha$.
This can summarized as follows: 
\begin{align}
    T &= \begin{bmatrix} \begin{array}{c|cc}
        -\lambda & \lambda p \alpha_1 & \lambda (1-p) \alpha_2 \\
        \hline
        s_1^* & S_1 & 0 \\
        s_2^* & 0 & S_2
    \end{array} \end{bmatrix}, \hspace*{5mm}
    T^* = \begin{bmatrix} \begin{array}{c|c}
        -\lambda & \lambda \alpha \\
    \end{array} \end{bmatrix}, \hspace*{5mm}  \mbox{and} \hspace*{5mm}
    P = \begin{bmatrix} \begin{array}{c|c}
        1 & \vecn
    \end{array} \end{bmatrix}
\end{align}
In this particular case the solution $\Psi$ to  \eqref{eq:Ricat}
is a column vector of ones as $|\mathcal{S}^-| = 1$ and $\lambda < 1$. This implies that $K = T_{++} + \Psi T_{-+} = S + \lambda \e \alpha$, which is the $T$ matrix of the workload process, and $\pi_+ = \lambda (1-\lambda) \alpha = \lambda \beta$ so that $P[W^{(1)} > t] = \lambda \beta e^{T t} (-T)^{-1} \e = P[Z > t]$.

\subsection{Nudge-1 scheduling:}
In the second example, we show how a MMFQ can be used to model the virtual waiting time of a type-1 job in Nudge-$1$, which is Nudge$^1(K)$ with $K = 1$, but also Nudge*$(M)$ with $M = 1$. In this case, an arriving type-1 job may skip ahead of a single type-2 job at the end of the queue, but only if the type-2 job has not been swapped before. Such a type-2 job should not have its workload added to the fluid yet as it does not add to the virtual waiting time of an arriving type-1 job. This indicates that we need more than one state in $\mathcal{S}^-$ as we have to remember whether such a type-2 job exists at the back of the queue, the workload of which has not yet been added to the fluid. 
More specifically $\mathcal{S}^-$ now has two states: one state for when there is no type-2 job at the back of the queue that has not been swapped (and workload of which has not been added), which we denote as $0$, and one state for when there is such a 
type-2 job, which we denote as $1$. $\mathcal{S}^0$ still has only one state denoted as $0$, as it is only possible for the fluid to remain at zero when there is no type-2 job in the state. When a new arrival occurs, one of the following $4$ cases occurs:
\begin{enumerate}
    \item If the current state is $0 \in \mathcal{S}^-$ and the arrival is type-2, the state remains within
    $\mathcal{S}^-$ and simply changes from $0$ to $1$. So contrary to the FCFS case, we refrain from adding the
    work to the fluid at this stage as the type-2 job can potentially be passed by a later type-1 arrival.
    \item If the current state is $0 \in \mathcal{S}^-$ and the arrival is type-1, the state transitions to
    $\mathcal{S}^+$ to add the workload of the arriving type-1 job. After adding this work, the state
    returns to $0 \in \mathcal{S}^-$.
    \item If the current state is $1 \in \mathcal{S}^-$ and a type-2 arrival occurs, we end up with a queue ending with two type-2 jobs that have not been swapped. As an incoming type-1 job can only pass one of these, we do not need to remember both. We therefore add the workload of one type-2 job to the fluid before returning to state $1$. 
    In other words, in this case the state changes to $\mathcal{S}^+$ to add the work of a type-2 job and then returns
    to state $1 \in \mathcal{S}^-$.
    \item Finally, if the current state is $1 \in \mathcal{S}^-$ and a type-1 arrival occurs, this type-1 job passes the type-2 job stored in the state $1 \in \mathcal{S}^-$. As a result it can no longer be swapped again and its workload should be added to the fluid, as well as the type-1 job's workload. After adding the workload of both
    jobs the state should transition back to $0 \in \mathcal{S}^-$.
\end{enumerate}
To accommodate these $4$ types of events, the set $\mathcal{S}^+$ is partitioned into four subsets. 
The first subset covers case (2) and adds the workload of a type-1 job before returning to state $0 \in \mathcal{S}^-$. The second subset covers case (4) and first adds a type-2 job, then transitions to the first set to add a type-1 job before returning to state $0 \in \mathcal{S}^-$. The third and fourth subset add a type-2 job after which they transition to state $1 \in \mathcal{S}^-$ and  $0\in \mathcal{S}^- $, respectively. The third subset covers case (3), so the fourth subset may seem redundant at this point. However, its use will become clear when discussing the matrix $T^*$ and $P$. 
This results in the following $T$ matrix:
\begin{align}
        T &= \begin{bmatrix} \begin{array}{cc|cccc}
        -\lambda & \lambda (1-p) & \lambda p \alpha_1 & \vecn & \vecn & \vecn \\
        0 & -\lambda & \vecn & \lambda p \alpha_2 & \lambda (1-p) \alpha_2 & \vecn \\
        \hline
        s_1^* & \vecn & S_1 & 0 & 0 & 0 \\
        \vecn & \vecn & s_2^* \alpha_1 & S_2 & 0 & 0 \\
        \vecn & s_2^* & 0 & 0 & S_2 & 0 \\
        s_2^* & \vecn & 0 & 0 & 0 & S_2
    \end{array} \end{bmatrix} 
\end{align}
In the case where the fluid reaches zero from state $1 \in \mathcal{S}^- $, the stored type-2 job goes into service and so its workload must be added to the fluid. After the work has been added, the state is changed to $0\in \mathcal{S}^-$, this explains the use of the fourth subset of $\mathcal{S}^+$ introduced above. This indicates that
$P_{-+}$ is nonzero. While the fluid is zero (which can only happen in state $0 \in \mathcal{S}^0$), any incoming job's workload is added to the fluid as it immediately goes into service, i.e. no type-2 jobs are added to the state. This yields
the following expressions for $T^*$ and $P$:
\begin{align}
    T^* &= \begin{bmatrix} \begin{array}{c|cccc}
        -\lambda & \lambda p \alpha_1 & \vecn & \vecn & \lambda (1-p) \alpha_2 %\\
       %0 & -1 & \vecn & \vecn & \vecn & \vecn
    \end{array} \end{bmatrix} \hspace*{5mm} \mbox{and} \hspace*{5mm}
    P = \begin{bmatrix} \begin{array}{c|cccc}
        1 & \vecn & \vecn & \vecn & \vecn \\
        0 & \vecn & \vecn & \vecn & \alpha_2
    \end{array} \end{bmatrix}
\end{align}

\subsection{Nudge*$(M)$ scheduling:}

This whole idea can then expanded for Nudge*$(M)$. As before, the state keeps track of type-2 jobs at the end of the queue which have not been added to the fluid thus far, as these are the jobs that would be passed by an arriving type-1 job. While for Nudge-1 the state was a number $k \in \{0, 1\}$ representing whether or not the queue ends with a type-2 job that can be passed, this now becomes an $M$-dimensional vector of zeroes and ones. For (at most) the last $M$ arrivals, we need to remember which ones were type-2 jobs as well as their positions among the last $M$ arrivals.
For instance, if $M=6$ and the state equals $(1,1,0,1,0,0)$, it means the last four arrivals were (from newest to oldest) type-2, type-2, type-1 and type-2 and the workload of the three type-2 jobs is not yet added to the fluid. 
The zeros in the back of the state, being the last two zeros in our example, do not imply that the fifth and
sixth last arrivals were type-1. For instance, it could also be that the fifth last arrival was a type-2 job that is in service or that already left the system.

As such the state $s = (s_1, s_2, \dots, s_{M}) \in \{0, 1\}^M$ is an $M$-dimensional vector where entry $i$ is one if and only if the $i$-th last arrival was a type-2 job that is still waiting in the queue, meaning its workload was not yet added to the fluid. These states make up the set $\mathcal{H} = \{0, 1\}^M$.

Define two functions to manipulate states:
$$dec: \mathcal{H} \rightarrow \mathcal{H}$$
$$shift: \mathcal{H} \times \{0, 1\} \rightarrow \mathcal{H}$$
$dec$ removes the oldest type-2 job that is being remembered in the state: it decrements the rightmost non-zero entry of $s$. For instance, $dec((1,1,0,1,0,0))  = (1,1,0,0,0,0)$.
$shift$ is used when an arrival occurs and moves all the waiting type-2 jobs stored in the state and adds a $0$
or $1$ in front: $shift(s, v) = (v, s_1, s_2, \dots, s_{M-1})$ such that $s_M$ is dropped from the state.

The set of phases in $\mathcal{S}^-$ is therefore denoted as
$$\mathcal{S}^- = \{ \langle s, 0 \rangle \hspace{2mm} | \hspace{2mm} s \in \mathcal{H} \}.$$
As in Nudge-$1$, the fluid can only remain at zero from one state which denotes that there are no type-2 jobs in the last $M$ arrivals that are still waiting in the queue: $\mathcal{S}^0 = \{ \langle \vecn, 0 \rangle \} \subseteq \mathcal{S}^-$.

When a new arrival occurs in state $\langle s, 0 \rangle$ with $s = (s_1, \dots, s_M)$, one of the following $4$ cases occurs, depending on $s_M$ and the type of the new arrival. These $4$ cases are similar to those for Nudge-$1$. \begin{enumerate}
    \item If $s_M = 0$ and the arrival is type-2, the state simply changes from $\langle s, 0 \rangle$ to $\langle shift(s, 1), 0 \rangle$. This change indicates that the arrivals stored in the state are shifted over by a new type-2 job. The requirement $s_M = 0$ means that this new arrival does not move a stored type-2 job out of the last $M$, as in such case its workload must be added to the fluid.
    \item If $s_M = 0$ and the arrival is type-1, the workload of the arriving type-1 job is added
    via a set of states that we will denote as $\mathcal{S}^+_1$. After adding this work, the state
    returns to $\langle shift(s, 0), 0 \rangle$ as the stored arrivals are shifted over by a new type-1 job.
    \item If $s_M = 1$ and a type-2 arrival occurs, we end up with a queue that includes a type-2 job followed by $M$ arrivals. This type-2 job was stored in $s$, but can no longer be passed by arriving type-1 jobs. As a result, we add the workload of one type-2 job (via a set of states defined further as $\mathcal{S}_3^+$) 
    to the fluid before returning to state $\langle shift(s, 1), 0 \rangle$, which no longer includes this type-2 job.
    %In other words, in this case the state changes to $\mathcal{S}^+$ to add the work of a type-2 job and then returns to state $1 \in \mathcal{S}^-$.
    \item Finally, if $s_M = 1$ and a type-1 arrival occurs, we again have a type-2 job that is now safe, like in the previous case. This means the workload of one type-2 job must be added to the fluid, as well as the type-1 job's workload. Afterwards, the state changes to $\langle shift(s, 0), 0 \rangle$.
    %this type-1 job passes the type-2 job stored in the state $1 \in \mathcal{S}^-$. As a result it can no longer be swapped again and its workload should be added to the fluid, as well as the type-1 job's workload. After adding the workload of both jobs the state should transition back to $0 \in \mathcal{S}^-$.
\end{enumerate}
To accommodate these $4$ types of events, the set $\mathcal{S}^+$ is partitioned into three subsets
$\mathcal{S}^+_1$,$\mathcal{S}^+_2$ and $\mathcal{S}^+_3$. %These are similar to the four subsets in the Nudge-$1$ example.
%The first subset covers case (2) and adds the workload of a type-1 job before returning to state $0 \in \mathcal{S}^-$. The second subset covers case (4) and first adds a type-2 job, then transitions to the first set to add a type-1 job before returning to state $0 \in \mathcal{S}^-$. The third and fourth subset add a type-2 job after which they transition to state $1 \in \mathcal{S}^-$ and  $0\in \mathcal{S}^- $, respectively. The third subset covers case (3),

%Three sets of phases are used for upwards jumps: these correspond to adding the workload of jobs to the fluid instantaneously. %, and are similar to those for Nudge$^1(K)$.
The first subset
$$\mathcal{S}_1^+ = \{ \langle s, i, 1 \rangle \hspace{2mm} | \hspace{2mm} s \in \mathcal{H} \land s_1 = 0 \land i \in \{1, \dots, n_1\} \}$$
adds the workload of a type-1 job to the fluid, transitioning to state $\langle s, 0 \rangle \in \mathcal{S}^-$ afterwards. This covers case (2) and happens when a type-1 job has arrived, and so there is a  type-1 job at the end of the queue, explaining the extra condition $s_1 = 0$. 

The second subset
$$\mathcal{S}_2^+ = \{ \langle s, i, 2 \rangle \hspace{2mm} | \hspace{2mm} s \in \mathcal{H} \land s_1 = 0 \land i \in \{1, \dots, n_2\} \}$$
is used to cover case (4), where the workload of a type-2 job followed by the workload of a type-1 job is added
to the fluid. 
The set  $\mathcal{S}_2^+$ is used to add the type-2 work, then a transition is made to the
set $\mathcal{S}_1^+$ to add the work of the type-1 job, before finally transitioning to state $\langle s, 0 \rangle \in \mathcal{S}^-$. This set has the same requirement that $s_1 = 0$ as the set $\mathcal{S}_1^+$ because case (4) also involves a type-1 arrival.

Finally the third subset
$$\mathcal{S}_3^+ = \{ \langle s, i, 3 \rangle \hspace{2mm} | \hspace{2mm} s \in \mathcal{H} \land i \in \{1, \dots, n_2\} \}$$ adds the workload of a type-2 job to the fluid and then transitions to $\langle s, 0 \rangle \in \mathcal{S}^-$. This subset covers case (3).

These three subsets are similar to the four subsets in the Nudge-$1$ example. In particular, subsets 1 and 2 have the same purpose as the first two subsets in the Nudge-1 example, while $\mathcal{S}_3^+$ can be regarded as the
union of the third and fourth subset, as this set adds the workload of type-2 jobs to the fluid before transitioning to some state in $\mathcal{S}^-$.

We now define all the necessary transitions between states, based on the four cases when an arrival occurs. %Let $s \xrightarrow{r} s'$ denote a transition from $s$ to $s'$ with rate $r$ or probability $r$ depending on the context. This $r$ may be a vector or a matrix if $s$, $s'$, or both are states in $\mathcal{S}^+$ in which case all $n_1$ or $n_2$ phases are covered. For example, $\langle s, 0 \rangle \xrightarrow{\lambda p \alpha_2} \langle s', \cdot, 2 \rangle$ is shorthand for $\langle s, 0 \rangle \xrightarrow{(\lambda p \alpha_2)_i} \langle s', i, 2 \rangle$ for  $i \in \{1, \dots, n_2\}$, while $\langle s,\cdot, 2 \rangle \xrightarrow{s_2^* \alpha_1} \langle s', \cdot , 1\rangle$ is shorthand for $\langle s, i, 2 \rangle \xrightarrow{(s_2^* \alpha_1)_{i,j}} \langle s', j, 1 \rangle$ with $i \in \{1, \dots, n_2\}, j \in \{1, \dots, n_1\}$.
%Entries of these matrices are 0 if no transition is specified, and they are partitioned into different parts depending on the four sets of states ($\mathcal{S}^-$, $\mathcal{S}^+_1$, $\mathcal{S}^+_2$ and $\mathcal{S}^+_3$). 
Transitions from $\mathcal{S}^-$ are for the four cases that were previously described, depending on the type of the arrival and $s_M$. %to $\mathcal{S}^-$ are when a new type-2 job arrives: depending on $s_M$, a new type-2 job is added to the state without making any jumps, or a type-2 job's workload has to be added to the fluid via the set $\mathcal{S}^+_3$.
\begin{align}
    (T_{--})_{\langle s, 0 \rangle, \langle s', 0 \rangle}
    &= \left\{ 
    \begin{array}{ll}
    -\lambda & \quad s' = s, \\
    \lambda (1-p) & \quad s_M = 0, s' = shift(s, 1), \\
    0 & \quad otherwise.
    \end{array} \right. \\
    (T_{-+})_{\langle s, 0 \rangle, \langle s', j, k \rangle} &=
    \left\{\begin{array}{ll}
    (\lambda p \alpha_1)_j & \quad s_M = 0, s' = shift(s, 0), k = 1, \\
    (\lambda p \alpha_2)_j & \quad s_M = 1, s' = shift(s, 0), k = 2, \\
    (\lambda (1-p) \alpha_2)_j & \quad s_M = 1, s' = shift(s, 1), k = 3, \\
    0 & \quad otherwise.
    \end{array} \right.
\end{align}

In $\mathcal{S}_1^+$, a type-1 job's workload is added to the fluid before returning to $\mathcal{S}^-$. This gives two sets of transitions:

\begin{align}
    (T_{1+})_{\langle s, i, 1 \rangle, \langle s', j, k \rangle}
    &= \left\{ 
    \begin{array}{ll}
    (S_1)_{i, j} & \quad s' = s, k = 1, \\
    0 & \quad otherwise.
    \end{array} \right. \\
    (T_{1-})_{\langle s, i, 1 \rangle, \langle s', 0 \rangle} &=
    \left\{\begin{array}{ll}
    (s_1^*)_i & \quad s' = s, \\
    0 & \quad otherwise.
    \end{array} \right.
\end{align}

In $\mathcal{S}_2^+$, a type-2 job is added to the fluid after which the state moves to $\mathcal{S}_1^+$ to add one type-1 job.

\begin{align}
    (T_{2+})_{\langle s, i, 2 \rangle, \langle s', j, k \rangle}
    &= \left\{ 
    \begin{array}{ll}
    (s_2^* \alpha_1)_{i, j} & \quad s' = s, k = 1, \\
    (S_2)_{i, j} & \quad s' = s, k = 2, \\
    0 & \quad otherwise.
    \end{array} \right. \\
    (T_{2-})_{\langle s, i, 2 \rangle, \langle s', 0 \rangle} &= 0
\end{align}

In $\mathcal{S}_3^+$, only a type-2 job is added:

\begin{align}
    (T_{3+})_{\langle s, i, 3 \rangle, \langle s', j, k \rangle}
    &= \left\{ 
    \begin{array}{ll}
    (S_2)_{i, j} & \quad s' = s, k = 3, \\
    0 & \quad otherwise.
    \end{array} \right. \\
    (T_{3-})_{\langle s, i, 3 \rangle, \langle s', 0 \rangle} &=
    \left\{\begin{array}{ll}
    (s_2^*)_i & \quad s' = s, \\
    0 & \quad otherwise.
    \end{array} \right.
\end{align}

We also need to define $T^*$, which gives transitions when the fluid level is exactly zero. This can only happen in state $\langle \vecn, 0 \rangle \in \mathcal{S}^0$. With fluid level 0, any arriving job is immediately added to the fluid. There's also the matrix $P$ which gives transition probabilities to immediately move to a different state when reaching fluid level 0. In state $\langle \vecn, 0 \rangle$ no transition is made, while in other states the oldest type-2 job is added to the fluid and removed from the state.
\begin{align}
    (T_{00}^*)_{\langle \vecn, 0 \rangle, \langle \vecn, 0 \rangle}
    &= -\lambda \\
    (T_{0+}^*)_{\langle \vecn, 0 \rangle, \langle s', j, k \rangle} &=
    \left\{\begin{array}{ll}
    (\lambda p \alpha_1)_j & \quad s' = \vecn, k = 1, \\
    (\lambda (1-p) \alpha_2)_j & \quad s' = \vecn, k = 3, \\
    0 & \quad otherwise.
    \end{array} \right. \\
    (P_{-0})_{\langle s, 0 \rangle, \langle \vecn, 0 \rangle} &=
    \left\{\begin{array}{ll}
    1 & \quad s = \vecn, \\
    0 & \quad otherwise.
    \end{array} \right. \\
    (P_{-+})_{\langle s, 0 \rangle, \langle s', j, k \rangle} &=
    \left\{\begin{array}{ll}
    (\alpha_2)_j & \quad s \neq \vecn, s' = dec(s), k = 3, \\
    0 & \quad otherwise.
    \end{array} \right.
\end{align}
With the definition of the matrices $T$, $T^*$ and $P$, the stationary distribution of the fluid can be calculated as shown before, yielding the density $f_{W^{(1)}}(t)$ and complementary distribution function $P[W^{(1)} > t]$ of the waiting time of a type-1 job using $\pi_+$ and the matrices $K$ and $\Psi$.

\begin{theorem}
    Let $W^{(1)}$ be the waiting time distribution of a type-1 job, then
    \begin{align}
        P[W^{(1)} > t] = \pi_+ e^{Kt} (-K)^{-1} \Psi \e
    \end{align}
    
    Let $R^{(1)}$ be the response time distribution of a type-1 job, then
    \begin{align}
        P[R^{(1)} > t] = P[W^{(1)} > t] + (1-\lambda) \alpha_1 e^{S_1 t} \e + (\pi_+, 0) e^{B t} \begin{bmatrix} \vecn \\ \e \end{bmatrix},
    \end{align}
    where
    \begin{align*}
        %v_1 &= (\pi_+, 0) \\
        B &= \begin{bmatrix}
            K & \Psi \e \alpha_1 \\
            0 & S_1
        \end{bmatrix} .
        %v_2 &= \begin{bmatrix} \vecn \\ \e_{n_1} \end{bmatrix}
    \end{align*}
\end{theorem}
\begin{proof}{Proof:}
    The formula for the waiting time comes from \eqref{eq:fluidqueue} as we constructed the fluid queue in such a way that its fluid level is the virtual waiting time of a type-1 job.

    A type-1 job's response time consists of its waiting time combined with its service time. The response time is larger than $t$ if the waiting time is. If the waiting time is $s < t$, the response time may still be larger than $t$ if the job's service time is longer than $t-s$. The waiting time is 0 if the queue is empty with probability $1-\lambda$.
    \begin{align*}
        P[R^{(1)} > t] &= P[W^{(1)} > t] + (1-\lambda) P[X_1 > t] + \int_{0}^{t} f_{W^{(1)}}(s) P[X_1 > t-s] \,ds \\
        &= P[W^{(1)} > t] + (1-\lambda) \alpha_1 e^{S_1 t} \e + \int_{0}^{t} \pi_+ e^{Ks} \Psi \e \alpha_1 e^{S_1 (t-s)} \e \,ds % \\
        %&= P[W^{(1)} > t] + (1-\lambda) \alpha_1 e^{S_1 t} \e + \pi_+ \left( \int_{0}^{t} e^{Ks} \Psi \e \alpha_1 e^{S_1 (t-s)} \,ds \right) \e
    \end{align*}
    This gives the desired formula using Lemma 1 from \citep{vanhoudt_nudge}.
\Halmos \end{proof}

\bibliographystyle{informs2014}
\bibliography{thesis}

%\end{APPENDIX}
\end{document}